\documentclass[conference]{IEEEtran}
\IEEEoverridecommandlockouts
\usepackage{cite}
\usepackage{amsmath,amssymb,amsfonts}
\usepackage{algorithmic}
\usepackage{graphicx}
\usepackage{textcomp}
\usepackage{xcolor}
\usepackage[caption=false]{subfig}
\def\BibTeX{{\rm B\kern-.05em{\sc i\kern-.025em b}\kern-.08em
    T\kern-.1667em\lower.7ex\hbox{E}\kern-.125emX}}

\DeclareMathAlphabet{\mathpzc}{OT1}{pzc}{m}{it}
\usepackage{graphicx,subfloat}
\usepackage{balance}  
\usepackage{booktabs} 
\usepackage{amsmath}
\usepackage{enumerate}
\usepackage[shortlabels]{enumitem}
\usepackage{xspace}
\usepackage{mathrsfs}
\usepackage{tikz}
\usepackage{multirow}
\usepackage{weiwAlgorithm}
\usepackage[nice]{nicefrac}
\usepackage{array, makecell}
\usepackage{url}
\usepackage{diagbox}
\usepackage{textcomp}
\usepackage{arydshln}

\newcommand{\ourproblem}{WCSD}
\newcommand{\ourindex}{WC-INDEX}

\newcommand{\rev}[1]{\textcolor[rgb]{0,0,0}{#1}}

\newcommand{\ours}{WC-INDEX}
\newcommand{\ourp}{WC-INDEX+}

\renewcommand{\baselinestretch}{0.975}

\usepackage{enumitem}
\setlist{nolistsep}
\usepackage{amsmath}
\makeatletter
\g@addto@macro\normalsize{%
\setlength\abovedisplayskip{-1pt}
\setlength\belowdisplayskip{0pt}
\setlength\abovedisplayshortskip{-1pt}
\setlength\belowdisplayshortskip{0pt}
}
\makeatother

\newtheorem{definition}{Definition}
\newtheorem{lemma}{Lemma}
\newtheorem{theorem}{Theorem}
\newtheorem{example}{Example}

\newtheorem{observation}{Observation}
\newtheorem{proof}{\textbf{Proof}}

\long\def\comment#1{}
\newcommand{\stitle}[1]{\vspace{1ex}\noindent{{\bf#1}}}

\newcommand{\refalg}[1]{Algorithm~\ref{alg:#1}}

\newcommand{\topcaption}{%
	\setlength{\abovecaptionskip}{0.01cm}%
	\setlength{\belowcaptionskip}{0.01cm}%
	\caption}

\begin{document}
\title{Efficiently Answering Quality Constrained Shortest Distance Queries in Large Graphs\\
\thanks{Xuemin Lin is the corresponding author.}}


\author{{You Peng$^{\dagger}$, Zhuo Ma$^{\dagger}$, Wenjie Zhang$^{\dagger}$, Xuemin Lin$^{\ddagger}$, Ying Zhang$^{\S}$, Xiaoshuang Chen$^{\dagger\dagger}$
} %
\vspace{1.6mm}\\
\fontsize{10}{10}
\selectfont\itshape
$^\dagger$The University Of New South Wales, \\
$^{\ddagger}$Antai College of Economics \& Management, Shanghai Jiao Tong University, Shanghai, China,\\
$^\S$QCIS, University of Technology, Sydney\\
$^{\dagger\dagger}$Data Principles (Beijing) Technology Co., Ltd.\\
\fontsize{9}{9} \selectfont\ttfamily\upshape
unswpy@gmail.com, zhuo.ma@student.unsw.edu.au, wenjie.zhang@unsw.edu.au, \\
xuemin.lin@sjtu.edu.cn, Ying.Zhang@uts.edu.au, xiaoshuang.chen@enmotech.com\\}

\maketitle

\begin{abstract}
The shortest-path distance is a fundamental concept in graph analytics and has been extensively studied in the literature. In many real-world applications, quality constraints are naturally associated with edges in the graphs and finding the shortest distance between two vertices $s$ and $t$ along only valid edges (i.e., edges that satisfy a given quality constraint) is also critical. In this paper, we investigate this novel and important problem of quality constraint shortest distance queries. We propose an efficient index structure based on 2-hop labeling approaches. Supported by a path dominance relationship incorporating both quality and length information, we demonstrate the minimal property of the new index. An efficient query processing algorithm is also developed. Extensive experimental studies over real-life datasets demonstrates efficiency and effectiveness of our techniques.

\end{abstract}




\section{Introduction}
\label{sect:intro}
Shortest-path distance is a critical concept in graph analytics~\cite{10.1145/3448016.3457237,jin2021fast,peng2021efficient,peng2021dlq,peng2021answering}. Specially, a path between two vertices $s$ and $t$ is a shortest path if its length is the shortest among all paths between $s$ and $t$. The distance of the shortest path is called the shortest-path distance, or shortest distance for short. Due to its optimality, the notion of the shortest distance has been exploited to tackle a broad range of problems, including keyword search~\cite{he2007blinks,jiang2015exact,tao2011nearest}, betweenness centrality~\cite{brandes2001faster,puzis2007fast} and route planning~\cite{abraham2011hub,abraham2012hierarchical,peng2020answering,qiu2018real,peng2018efficient,peng2019towards,lai2021pefp}. 
The shortest distance between two vertices $s$ and $t$ can reflect the vertices' significance. For instance, i) in the nearest keyword search, the vertices closest to the query source are favoured~\cite{jiang2015exact}; and ii) in social networks, distances are employed in the search ranking to aid users in identifying the most relevant results~\cite{vieira2007efficient}.

Many real-world networks~\cite{DBLP:journals/pvldb/HaoYZ21,DBLP:conf/adc/LiHYCZ22} naturally impose a quality constraint over edges. For instance, in a road network, road segments may specify the weight limits permitted for auto-trucks. In this scenario, the weight limit is the quality of edges on road networks, and it is demanded to compute the shortest path by which an auto-truck can pass. Namely, compute the shortest path and the distance along which the auto-truck satisfies the quality constraint of each edge. This motivates us to formulate the quality constraint shortest distance problem: given query vertices $s$ and $t$ in a graph $G$ and a quality constraint $w$, the quality constraint shortest distance problem finds the shortest-path distance where the quality of each edge along the path is at least $w$. Note that while we present the techniques focusing on the shortest distance computation, we will show in Section \ref{sect:ext} that our method can also easily support quality constraint shortest path queries. 

\begin{figure}[htbp]
\centering
\includegraphics[scale=0.35]{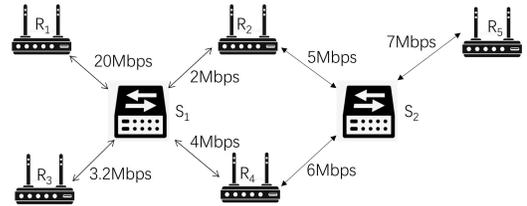}
\caption{A Communication Network. $R_i$ indicates the $i^{th}$ router, and $S_i$ indicates the $i^{th}$ switch.}
\label{fig:motivation_wcspd}
\end{figure}

\stitle{Applications.} Below we introduce some motivating applications. 

\noindent(1) \underline{\textit{Communication Networks~\cite{qiao2013computing}}}.~To achieve end-to-end Quality-of-Service (QoS) guarantees~\cite{ma1997path}, the transmission of multimedia streams imposes a minimum-bandwidth requirement on all the links of a path. A quality constrained shortest distance query can determine the distance (for the consideration of minimum cost or delay) between two nodes in a network, where each edge/link has a minimum bandwidth demand of $w$. The resultant path can handle $w$ bits per second for the transmission of a stream, such as audio or video, with guaranteed bandwidth. Figure~\ref{fig:motivation_wcspd} illustrates a motivating example. Given a minimum speed guarantee such as $3$ Mbps, a query asks for the distance from $R_3$ to $R_2$ with such a speed guarantee. In this case, the resultant distance is $4$ since $R_3 \rightarrow S_1 \rightarrow R_4 \rightarrow S_2 \rightarrow R_2$ fulfills all the criteria, while $R_3 \rightarrow S_1 \rightarrow R_2$ does not owing to the speed of $S_1 \rightarrow R_2 = 2$Mbps $<$ $3$Mbps.

\noindent(2) \underline{\textit{Social Networks~\cite{vieira2007efficient,qing2022towards,yuan2022efficient,peng2022finding,chen2022answering}}}.~In social networks, determining the closeness of two individuals is a critical issue. A popular metric is the distance between them in the social networks, e.g., a 2-hop friendship connection is stronger than a 3-hop one. The strength of connections between users is indicated based on profile similarity and interaction activity~\cite{xiang2010modeling,zhao2012relationship} and the distance between users needs to incorporate such strength information. To support this, a distance query with the quality constraint identifies the distance with only strong connections.

\noindent(3) \underline{\textit{Biology Networks~\cite{leser2005query}}}. Pathway queries are vital in the analysis of biological networks, where vertices represent the entities, e.g., enzymes and genes, while edges reflect interactions or relations~\cite{leser2005query}. As shown in~\cite{leser2005query}, one of the four important pathway queries in biological networks is to identify a shortest path between two substances subject to certain constraints. The quality can derive from the activity of kinase~\cite{kitagawa2013activity,biosa2013gtpase}. A frequently issued query in these biology networks is to determine the shortest pathway from substance $u$ to transfer to substance $v$, where all the activities of kinase on this pathway is as least $w$.

\stitle{Challenges.}~\rev{In real applications, the quality constraint shortest distance queries can be issued frequently over large-scale graphs. It requires both real-time response time and scalability}. An online BFS-based search needs to traverse the graph for given query vertices $s$ and $t$, making it impractical for real scenarios where real-time responses are demanded. 2-hop labelling approaches are shown to be efficient to support distance queries. Nevertheless, to deal with the constraints on edge qualities, a na\"ive adaption of 2-hop labelling solution involves constructing an index for every possible quality value $w$ among all edges of the graph. Such a solution is infeasible since the number of distinct $w$ values can be large. To overcome these challenges, a modified 2-hop labeling index is designed for the quality constrained distance problem to fill this research gap. To further accelerate index construction and query processing, and to reduce index size, this paper investigates various pruning methods, proposes a query-efficient approach, and develops efficient vertex ordering strategies. \rev{Since these applications deserve both real-time response time (the online method could not satisfy) and scalability (the existing index-based method could not satisfy), our proposed method could cope with all these two challenges.}

\noindent \rev{\textbf{Novelty}.~Our approach incorporates an extension to the 2-hop index that takes advantage of the ordering of weights and distance, and significantly prunes vertices that would have been processed in a classical 2-hop index while maintaining index minimality. We also investigated the ordering of BFS searches and discovered that using vertex degree or tree decomposition can have different effects on different kinds of networks.}

\stitle{Our Approach.}~To efficiently answer the quality constrained shortest distance problem, this work develops a modified 2-hop labeling based index which possesses \textit{soundness}, \textit{completeness}, and \textit{minimal} properties. We investigate the BFS search orders in building the index, and propose a quality- and distance-priority constrained BFS to naturally meet the three properties without incurring additional costs. The query operation over the index is used not only in answering the quality constraint distance queries but also in the index construction phase. We carefully design the query function and achieve linear time complexity by utilizing a nice dominance property of the problem. Last, a hybrid vertex ordering is proposed to tackle both graphs with small and non-small treewidth. 

\stitle{Contributions.}~Our principal contributions are as follows:
\begin{itemize}
    \item \textit{Theoretical Analysis}.~First, the quality constrained distance problem is defined, which has a variety of applications in road networks, social networks, and biological networks. This paper theoretically analyzes the time and space complexity of this problem. In addition, it investigates the \textit{soundness}, \textit{completeness}, and \textit{minimal} properties, and proposes a sophisticated index capable of naturally preserving these three desirable features.
    \item \textit{Efficient Index}.~We propose a $2$-hop labeling based index method. Both query-efficient method and distance-prioritized traversal strategy are presented to expedite index construction. With a nice property of this problem, the query function could be implemented in linear time, which could accelerate both query time and indexing time.
    Additionally, a hybrid vertex ordering is investigated. In addition, we investigate how to simply modify our index to support the quality constraint shortest path problem.
    \item \textit{Comprehensive Experiments.}~Compared to the baselines, our comprehensive experiments demonstrate the efficiency and effectiveness of our proposed method. 
\end{itemize}

\vspace{1mm}
\noindent
{\bf Roadmap.} The rest of the paper is organized as follows. Section~\ref{sect:pre} introduces some preliminaries and Section~\ref{sect:baselines} introduces baseline solutions. Our 2-hop labeling based method is proposed in Section~\ref{sect:index}. Section~\ref{sect:ext} investigates some extension cases, followed by empirical studies in Section~\ref{sect:exp}. Section~\ref{sect:related} surveys important related work. Section~\ref{sect:conclusion} concludes the paper.

\section{Preliminaries}
\label{sect:pre}

\subsection{Problem Definition}
Quality (\underline{$w$}) \underline{C}onstrained \underline{S}hortest \underline{D}istance (\ourproblem) is defined over an undirected unweighted graph $G(V,E, \Delta, \delta)$, where $V(G)$ denotes the set of vertices, $E(G)$ denotes the set of edges, $\Delta \subset \mathbb{R}$ is a set of real-valued qualities, and $\delta: E(G) \rightarrow \Delta$ is a function that assigns each edge $e \in E(G)$ to a real-valued quality $w \in \Delta$. For each vertex $u \in V(G)$, $N_G(u)=\{v|(u,v) \in E(G)\}$ denotes the set of neighbors of $u$, and $deg_G(u)$ denotes the degree of $u$, i.e., $deg_G(u) = |N_G(u)|$. A path $p$ from the vertex $s \in V(G)$ to the vertex $t \in V(G)$ is a sequence of vertices $\langle v_0 \rightarrow v_1 \rightarrow \cdots \rightarrow v_k\rangle$ such that $s=v_0$, $t=v_k$ and $(v_{i-1},v_{i})$ is an edge that belongs to $E(G)$ for $\forall i \in [k]$. The length of $p$, denoted by $len(p)$, is the number of edges included in the path $p$, i.e., $len(p)=k$. A path between $s$ and $t$ is the shortest if its length is no larger than any other path between $s$ and $t$, and the distance between $s$ and $t$ in $G$, denoted by $dist_G(s,t)$, is defined as the length of the shortest path between $s$ and $t$. \rev{Table \ref{tab:notation} provides a summary of the notations used in this paper.}

\begin{definition}
(\textsc{$w$-Path}) Given a graph $G$ and a threshold $w$, a $w$-path, denoted by $p_{w}$, is a path in $G$ such that each of its edge has a quality not smaller than $w$, i.e., $\forall e\in p_{w}$, $\delta(e) \geq w$. 
\end{definition}

\begin{definition}
(\textsc{$w$-Constrained Distance}) Given two vertices $s$ and $t$ in a graph $G$, and a threshold $w$, the $w$-constrained distance between $s$ and $t$, denoted by $dist^w_G(s,t)$, is the minimum length among all the $w$-paths between $s$ and $t$.
\end{definition}

For simplify, this paper focuses on the distance first. Once the distance is found, the quality constrained shortest path can be easily located, and this extension will be discussed in Section~\ref{sect:ext}.

\begin{definition}[\ourproblem]
Given two vertices $s$ and $t$ in a graph $G$, and a real-valued threshold $w$, the \ourproblem \ problem is to answer the $w$-constrained distance query, i.e., computing the $w$-constrained distance between $s$ and $t$.
\end{definition}

\begin{figure}[t]
\centering
\includegraphics[scale=0.4]{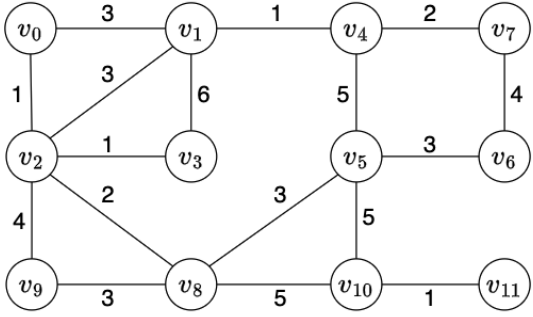}
\caption{An example graph. The values besides edges are their qualities.}
\label{fig:running_example_large}
\end{figure}

\begin{example}
Figure \ref{fig:running_example_large} depicts a weighted undirected graph, with the quality of each edge denoted by the number adjacent to it. In this example, a 1-constrained path between $v_0$ and $v_8$ is $\{ v_0 \rightarrow v_2 \rightarrow v_8\}$ since each edge on the path has a quality no less than $1$. It is also the shortest 1-constrained path between $v_0$ and $v_8$, therefore $dist_{1}(v_0,v_8)=2$. However, $\{v_0 \rightarrow v_2 \rightarrow v_8\}$ is not a 2-constrained path, since the edge $(v_0,v_2)$ has a quality less than 2. Alternatively, $\{v_0 \rightarrow v_1 \rightarrow v_2 \rightarrow v_8\}$ is the shortest 2-constrained path between $v_0$ and $v_8$, therefore $dist_{2}(v_0,v_8)=3$. For vertices $v_1$ and $v_4$, the path $\{v_1 \rightarrow v_2 \rightarrow v_9 \rightarrow v_8 \rightarrow v_5 \rightarrow v_4\}$ is both a 2-constrained path and a 3-constrained path. However, it is not the shortest 2-constrained path as $\{v_1 \rightarrow v_2 \rightarrow v_8 \rightarrow v_5 \rightarrow v_4\}$ also meets the constraint and has a shorter length. 
\end{example}

\subsection{2-Hop Labeling Framework}
\noindent\underline{\textit{Hub Labeling for Distance Queries}}.~Hub labeling~\cite{cohen2003reachability} is a vital category of algorithms for distance evaluation. In this class, a label $L(v)$ is computed for each vertex $v$ s.t. the distance between two vertices $s$ and $t$ can be obtained by inspecting $L(s)$ and $L(t)$ only, without traversing the graph. It is NP-hard to generate a labeling with the minimum size~\cite{cohen2003reachability}. Efficient hub labeling for road networks is explored in~\cite{abraham2011hub,abraham2012hierarchical}.~\cite{akiba2014fast} presents a labeling scheme that instead employs paths as hubs. Under the assumption of small treewidth and bounded tree height, ~\cite{ouyang2018hierarchy} proposed a scheme combining both hub labeling and hierarchy for road networks. Pruned landmark labeling (PLL)~\cite{akiba2013fast} is the state-of-the-art for real graphs, and its various extensions have been devised. For instance,~\cite{jiang2014hop} proposed an external algorithm that generates the same set of labels; ~\cite{li2019scaling} devised a parallel algorithm; and~\cite{akiba2014dynamic} describes an algorithm that updates the labels as new edges are inserted into the graph. ~\cite{li2017experimental} conducted an experimental study on hub labeling for distance queries.

\begin{table}[htb]
\centering
\caption{Summary of Notations}
\vspace{-2mm}
\begin{tabular}{c|l}
\noalign{\hrule height 1pt}
Notation & Definition \\ 
\noalign{\hrule height 0.6pt}
$G$ & a weighted unordered graph \\
$V(G)$ & a set of vertices\\
$E(G)$ & a set of edges\\
$\Delta$ & a set of real-valued qualities\\
$\delta$ & a function that assigns edges to a quality\\
$d$ & a distance value\\
$w$ & a quality value\\
$|w|$ & number of distinct values of edge qualities\\
$N_G(u)$ & neighbour set of $u$\\
$p$ & a path\\
$p_w$ & a $w$ path\\
$len(p)$ & number of edges included in path $p$\\
$dist_G(s,t)$ & distance between $s$ and $t$ in $G$\\
$dist^w_G(s,t)$ &  $w$-distance between $s$ and $t$ in $G$\\
$\mathcal{L}$ & WC-INDEX of $G$\\
$\mathcal{L}(u)$ & label set of $u$ in $\mathcal{L}$\\
$I$ & an index entry in the form of $(v,d,w)$\\
$R(u)$ & vector that records the current largest $w$ \\
       &of all paths from $v$ to all other vertices\\
\noalign{\hrule height 1pt}
\end{tabular}
\label{tab:notation}
\end{table}
\section{Baseline Solutions}
\label{sect:baselines}

\subsection{Basic Online and Indexing Approaches}
\noindent\underline{\textit{BFS-based Online Approaches}}.~A na\"ive online approach is to conduct a constrained breadth first search, which filters out-edges with quality values lower than the constraint $w$. The time and space complexity is both $O(|V|+|E|)$. Alternative algorithms such as Dijktra can also be performed. Another solution is to partition the original graph according to the values of quality, then it can perform constrained BFS on the corresponding partition. On large graphs, none of these online approaches is efficient in terms of query time. These algorithms are evaluated as baselines in the experiments.

\noindent\textbf{Details}.~Algorithm~\ref{alg:wc-bfs} depicts the BFS-based online computation. Line~\ref{wc_bfs:init} initializes the arrays $dis$ and $visited$. Line~\ref{wc_bfs:queue} initializes the search queue with the vertex $s$, and set $visited[s] = true$. Lines~\ref{wc_bfs:while1} to~\ref{wc_bfs:while2} constitute the search procedure. In each iteration, $size$ is set as the current size of $P$ and $dis = dis + 1$ in Line~\ref{wc_bfs:size}. Then, all the vertices are traversed in Line~\ref{wc_bfs:traver1} according the to vertex order. For vertex $u$ in $P$, each vertex $v$ in its neighbors is explored, and it will be pruned in Line~\ref{wc_bfs:prune1} if $e(u, v) < w$ or $visited[v]$. Line~\ref{wc_bfs:end} returns the $dis$ if explores $t$. Otherwise, $w$ is added into the $P$ and set $visited[v] = true$ to prevent duplication of candidates. $INF$ is returned in Line~\ref{wc_bfs:false} if not reach $t$.

\begin{algorithm}[thb]
	\SetAlgoVlined
	\SetFuncSty{textsf}
	\SetArgSty{textsf}
	\small
	\topcaption{WC-BFS} \label{alg:wc-bfs}
	\Input {any two vertices $s,t \in V$, and constraint $w$;}
	\Output {$dist^{w}$ between $s$ and $t$}
	\State{$dis = 0$, $\forall v \in V, visited[v] = false$}
	\label{wc_bfs:init}
    \State{$P.push((s))$, $visited[s] = true$;} \\
    \label{wc_bfs:queue}
    \While{$P \neq \emptyset$}{
    \label{wc_bfs:while1}
        \State{$size \leftarrow P.size()$, $dis = dis + 1$;} \\
        \label{wc_bfs:size}
        \For{$\forall i \in \{1,..,size\}$ }{
        \label{wc_bfs:traver1}
            \State{$(u) \leftarrow P.pop()$;} 
            \label{wcbfs:state1}\\
            \For{$\forall v \in adj[u]$}{
                \If{$e(u, v) < w$ or $visited[v]$}{
                \label{wc_bfs:prune1}
                \State{Continue;}\\
                }
                \If{$v == t$}{
                \label{wc_bfs:end}
                \State{Return $dis$;}\\
                }
                \State{$P.push(v)$, $visited[v] = true$;}\\
                \label{wc_bfs:add}
            }
        }
    }
    \label{wc_bfs:while2}
	\State{\Return $INF$;}\\
	\label{wc_bfs:false}
\end{algorithm}

\noindent\textit{{2-hop Labeling Approach}}.~2 hop-labeling approaches have been proven to be effective for addressing shortest path problems~\cite{ouyang2018hierarchy,akiba2013fast}. 2-hop indexes store the distance between vertices that has been precomputed. Each vertex $u \in V$ has its own label consisting of the form $(v, d)$, where $v$ is another vertex in the graph and $d$ is the distance from $u$ to $v$. To obtain the distance from vertex $s$ to vertex $t$, common vertices are identified in the labels of both $s$ and $t$, calculate the distance as the sum of their respective distance to $w$, and return the minimum sum of distances. 

\noindent\underline{\textit{Na\"ive 2-hop labeling method for \ourproblem}}.~A na\"ive 2-hop labeling solution involves filtering the graph based on edge qualities and constructing a classical 2-hop labeling index for each filtered graph. $|w|$ is used to denote the number of distinct values of edge qualities. Thus, $|w|$ 2-hop indices will be constructed, each containing only edges satisfying $\forall e \in E, \delta(e) \geq w$. Given a query $(s, t, w_0)$, it can be answered by using the classical 2-hop labeling method by a simple set intersection operation for the corresponding index with $w_0$. This approach becomes infeasible since the space required to store all of the indices grows as the graph sizes and $|w|$ increase. Moreover, in some instances, e.g., communication networks, the edge qualities are not integers. In such a scenario, it is impossible to create the na\"ive 2-hop labeling for every possible value of $w$. The time complexity of the na\"ive method is $O( |V| \cdot (|V| + |E|) \cdot |w| )$. The space complexity is $O(|V| \cdot |V| \cdot |w|)$ since the number of induced graphs is $|w|$, and in each of them, every vertex could store $|V|$ label entries in the worst case.



\section{Index Construction}
\label{sect:index}

\subsection{Our proposed 2-hop Labeling Index-based Approach}
The na\"ive approach can answer queries efficiently. Nevertheless, it needs to construct $|w|$ distinct 2-hop indices, which is inefficient and space-consuming. It becomes prohibitive to construct and maintain such a vast number of indices when $|w|$ is large. 

\begin{observation}
By building $|w|$ 2-hop labeling indexes, one may notice that numerous entries in the separate indices are redundant and elimination of those redundant entries does not impair the correctness of the queries. In this section, a modified 2-hop indexing approach is proposed which seeks to build only one index while efficiently supporting quality constrained shortest distance queries.
\end{observation}

Before providing this approach, the concept of path dominance is firstly described.

\begin{definition}
\label{def:dom}
(\textsc{Path Dominance}) Given two vertices $s$ and $t$ in a graph $G$, as well as two $w$-paths from $s$ to $t$, i.e., $p_{w_1}$ and $p_{w_2}$, $p_{w_1}$ dominates $p_{w_2}$ if $len(p_{w_1}) \leq len(p_{w_2})$ and $w_1 \geq w_2$.
\end{definition}
\smallskip

\begin{definition}
(\textsc{Minimal Path}) A $w$-path $p_{w}$ is a \textit{minimal} path if it cannot be dominated by any other $w$-path.
\end{definition}
\smallskip

\begin{example}
An example is illustrated in~Figure~\ref{fig:running_example}. For paths between vertices $v_0$ and $v_4$, path $\{v_0 \rightarrow v_3 \rightarrow v_4\}$ with length $2$ dominates path $\{v_0 \rightarrow v_3 \rightarrow v_5 \rightarrow v_4\}$ with length $3$, since the two paths have the same minimum edge quality of $1$ and the length of $\{v_0 \rightarrow v_3 \rightarrow v_4\}$ is smaller. For paths between vertices $v_1$ and $v_3$, $\{v_1 \rightarrow v_2 \rightarrow v_3\}$ with a minimum edge quality of $4$ dominates $\{v_1 \rightarrow v_0 \rightarrow v_3\}$ with a minimum edge quality of $1$, while both have the same length of $2$. Likewise, $\{v_1 \rightarrow v_3\}$ dominates $\{v_1 \rightarrow v_0 \rightarrow v_3\}$ due to both length and minimum edge quality. Path $\{v_0 \rightarrow v_3 \rightarrow v_4\}$ is the minimal $1$-path between $v_0$ and $v_4$, because it cannot be dominated by any other paths. Also, $\{v_1 \rightarrow v_2 \rightarrow v_3\}$ is both the minimal $3$-path and minimal $4$-path between $v_1$ and $v_3$.

\end{example}

In this paper, the dominance relationship between paths is leveraged and a single compact 2-hop index is generated, which is capable of answering queries regarding arbitrary quality constraint $w$. The \ourindex \ index is defined as follows:

\begin{definition}
(\textsc{\ourindex}) Given an undirected weighted graph $G$, a \ourindex \  $\mathcal{L}$ of $G$ assigns a label set $\mathcal{L}(u)$ to each vertex $u \in V(G)$. An index entry $(v, dist^{\bar{w}}_G(u,v), \bar{w}) \in \mathcal{L}(u)$ indicates that there exists a minimal $\bar{w}$-path between $u$ and $v$, and it also records the corresponding $\bar{w}$-constrained distance $dist^{\bar{w}}_G(u,v)$ between them.
\end{definition}


\begin{table}[thb]
    \centering
    \small
    \topcaption{WC-INDEX of Figure \ref{fig:running_example}} \label{tab:summary_of_social_networks}
    \scalebox{0.90}{
    \begin{tabular}{|l|l|} \hline
        \textbf{Vertex} & \textbf{$L(\cdot)$} \\ \hline
        $v_0$ & $(v_0,0,\infty)$\\ \hline
        $v_1$ & $(v_0,1,3),(v_1,0,\infty)$ \\ \hline
        $v_2$ & $(v_0,2,3),(v_1,1,5),(v_2,0,\infty)$ \\ \hline
        $v_3$ & \makecell[l]{$(v_0,1,1),(v_0,2,2),(v_0,3,3),(v_1,1,2),(v_1,2,4),(v_2,1,4),$ \\ $(v_3,0,\infty)$}\\ \hline
        $v_4$ & \makecell[l]{$(v_0,2,1),(v_0,3,2),(v_0,4,3),(v_1,2,2),(v_1,3,4),(v_2,2,4),$ \\ $(v_3,1,4),(v_4,0,\infty)$}\\ \hline
        $v_5$ & \makecell[l]{$(v_0,2,1),(v_0,3,2),(v_0,5,3),(v_1,2,2),(v_1,4,3),(v_2,2,2),$ \\ $(v_2,3,3),(v_3,1,2),(v_3,2,3),(v_4,1,3),(v_5,0,\infty)$}\\ \hline
    \end{tabular}}
\end{table}

\stitle{Query processing with \ourindex}. Given a complete WC-INDEX $\mathcal{L}$ of $G$, for any two vertices $s,t \in V(G)$ and an arbitrary real-value $w \in \mathbb{R}$, query $Q(s,t,w)$ computes the $w$-constrained distance between $s$ and $t$ as:

\begin{equation}
    dist^w_G (s,t)= \min_{\begin{subarray}{c}u \in \mathcal{L}(s)\cap \mathcal{L}(t) \\ w_1, w_2 \geq w \end{subarray}} dist^{w_1}_G(s,u) + dist^{w_2}_G(u,t) 
\end{equation}



\begin{algorithm}[thb]
	\SetAlgoVlined
	\SetFuncSty{textsf}
	\SetArgSty{textsf}
	\small
	\topcaption{Query Algorithm} \label{alg:query}
	\Input {any two vertices $s,t \in V$, and constraint $w$;}
	\Output {$dist^{w}$ between $s$ and $t$}
	\State{$dist^{w} \leftarrow \infty$;} \\
	\For {every index entry $I_i$ in $L(s)$} {
	    \If {$I_i.quality \geq w$} {
	        \For {every index entry $I_j$ in $L(t)$ such that $I_j.vertex = I_i.vertex$} {
	            \If {$I_j.quality \geq w$} {
	                \If {$I_i.dist + I_j.dist < dist^{w}$} {
	                    $ dist^{w} \leftarrow I_i.dist + I_j.dist$
	                }
	            }
	        }
	    }
	}
	\State{\Return $dist^{w}$;}\\ 
\end{algorithm}
\smallskip

\begin{figure}[t]
\centering
\includegraphics[scale=0.35]{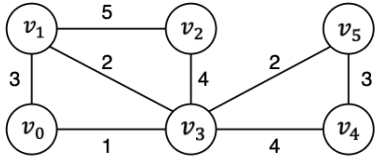}
\caption{A running example.}
\label{fig:running_example}
\end{figure}

\begin{example}
Figure~\ref{fig:running_example} illustrates how these 2-hop labeling index works. Given a query $Q(v_2,v_5,2)$, $L(v_2)$ and $L(v_5)$ are explored. This example starts with the first entry of $L(v_2)$, $(v_0,2,3)$, and discover that it satisfies the quality constraint of $2$. In the following, entries in $L(v_5)$ are explored that share the same vertex $v_0$ and also satisfy the quality constraint. $(v_0,3,2)$ is the first constraint-satisfying entry in $L(v_5)$. Therefore, $dist^2=2+3=5$ is obtained. The next entry $(v_0,4,3)$ also satisfies the constraint. Nevertheless, since the resultant distance $dist^2=2+4=6$ is larger than the previous distance obtained, no update is performed and $dist^2$ remains as $5$. It then moves on to the second entry of $L(v_2)$ which satisfies the constraint: $(v_1,1,5)$. In $L(v_5)$, label entries $(v_1,2,2)$ and $(v_1,4,3)$ are found satisfactory and subsequently update the distance as $dist^3=1+2=3$. Lastly, it visits $(v_2,0,\infty)$ in $L(v_2)$ and finds $(v_2,2,2)$ in $L(v_5)$, resulting in $dist^2=0+2=2$.
\end{example}

\begin{figure*}[htbp]
\subfloat[The $1^{st}$ iteration.\label{sfig:testa}]{%
  \includegraphics[width=.3\linewidth]{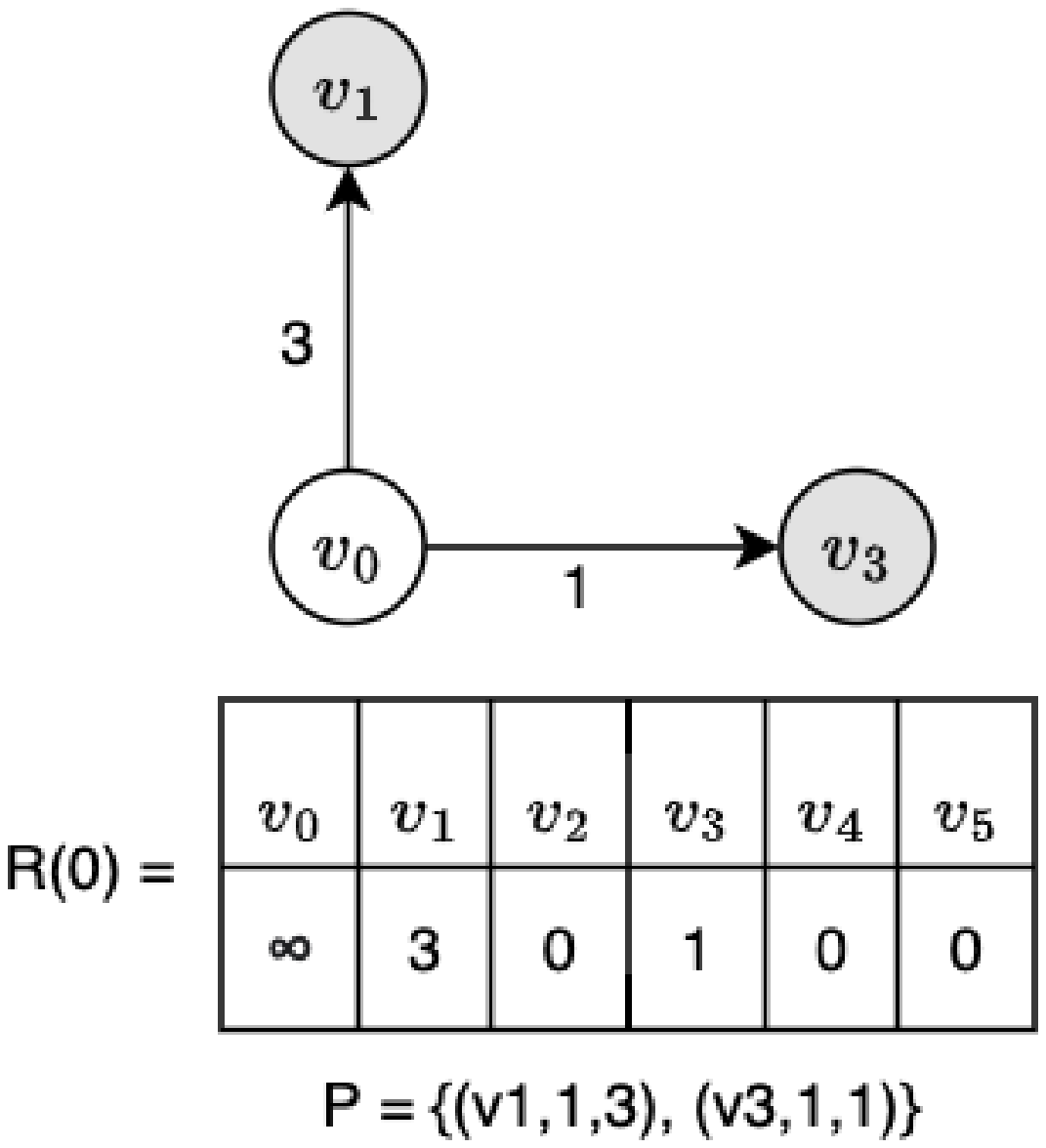}%
  \label{fig:bfs_a}
}\hfill
\subfloat[The $2^{nd}$ iteration.\label{sfig:testa}]{%
  \includegraphics[width=.3\linewidth]{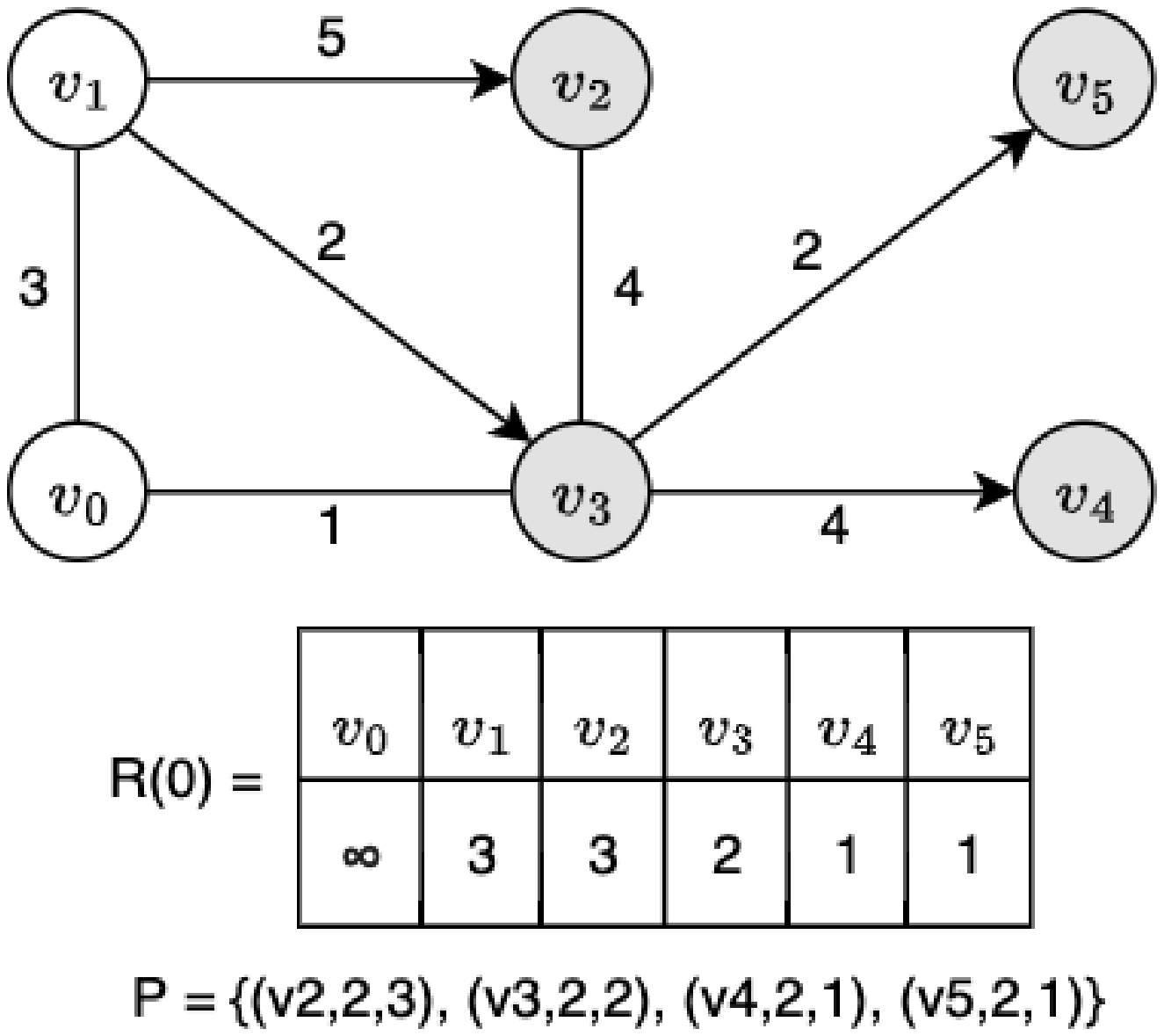}%
  \label{fig:bfs_b}
}\hfill
\subfloat[The $3^{rd}$ iteration.\label{sfig:testa}]{%
  \includegraphics[width=.3\linewidth]{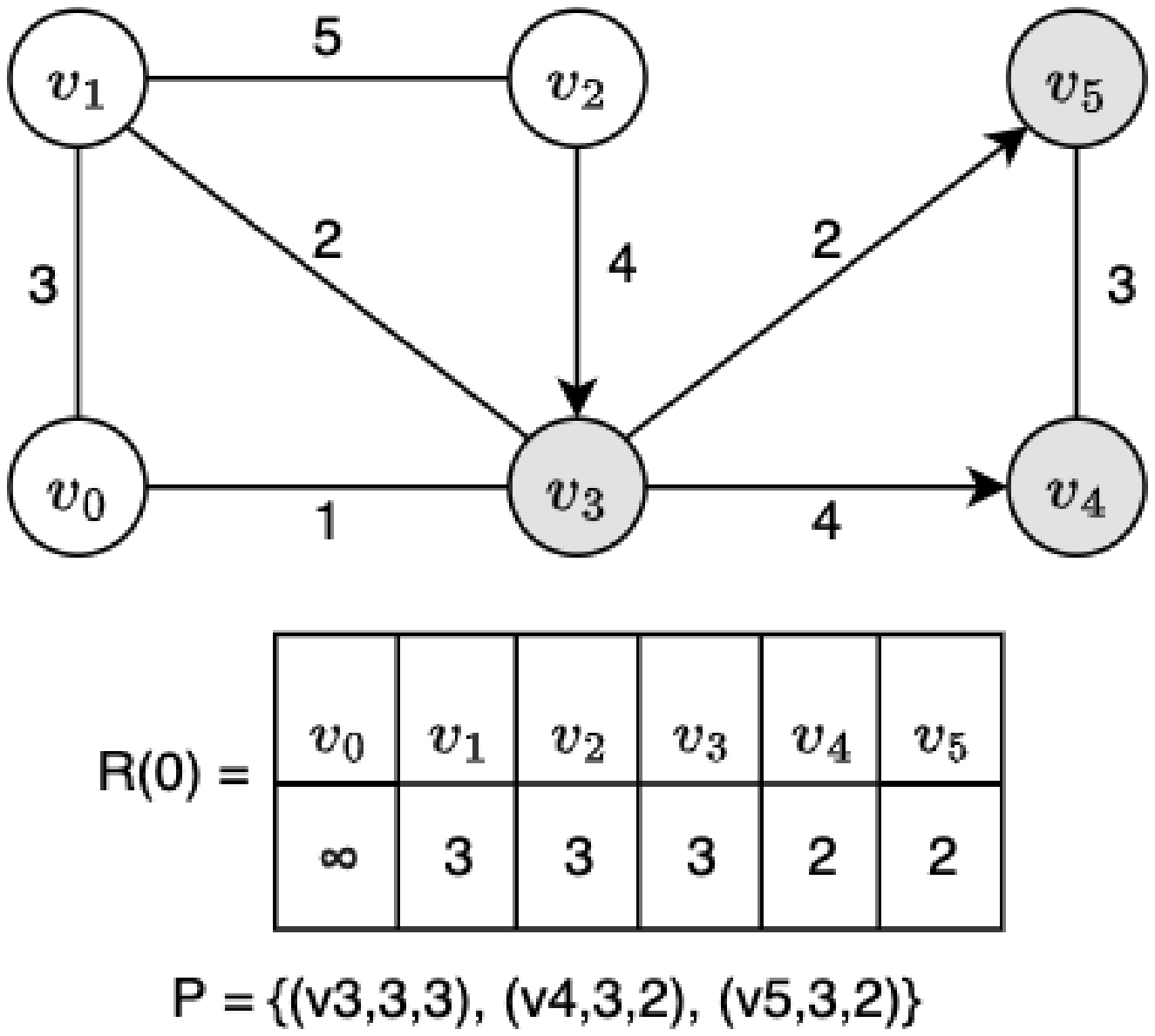}%
  \label{fig:bfs_c}
}\hfill
\subfloat[The $4^{th}$ iteration.\label{sfig:testa}]{%
  \includegraphics[width=.3\linewidth]{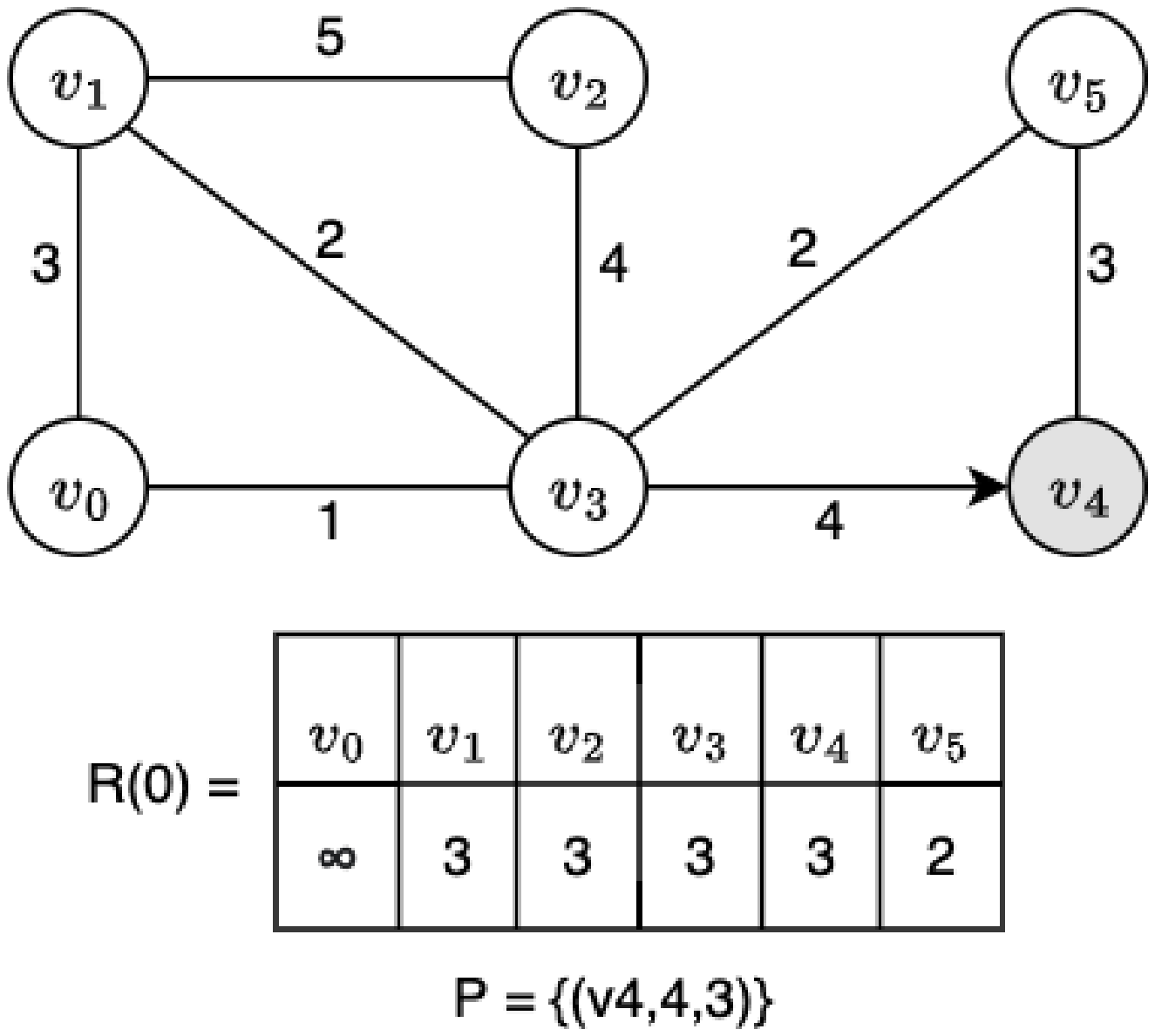}%
  \label{fig:bfs_d}
}\hfill
\subfloat[The $5^{th}$ iteration.\label{sfig:testa}]{%
  \includegraphics[width=.3\linewidth]{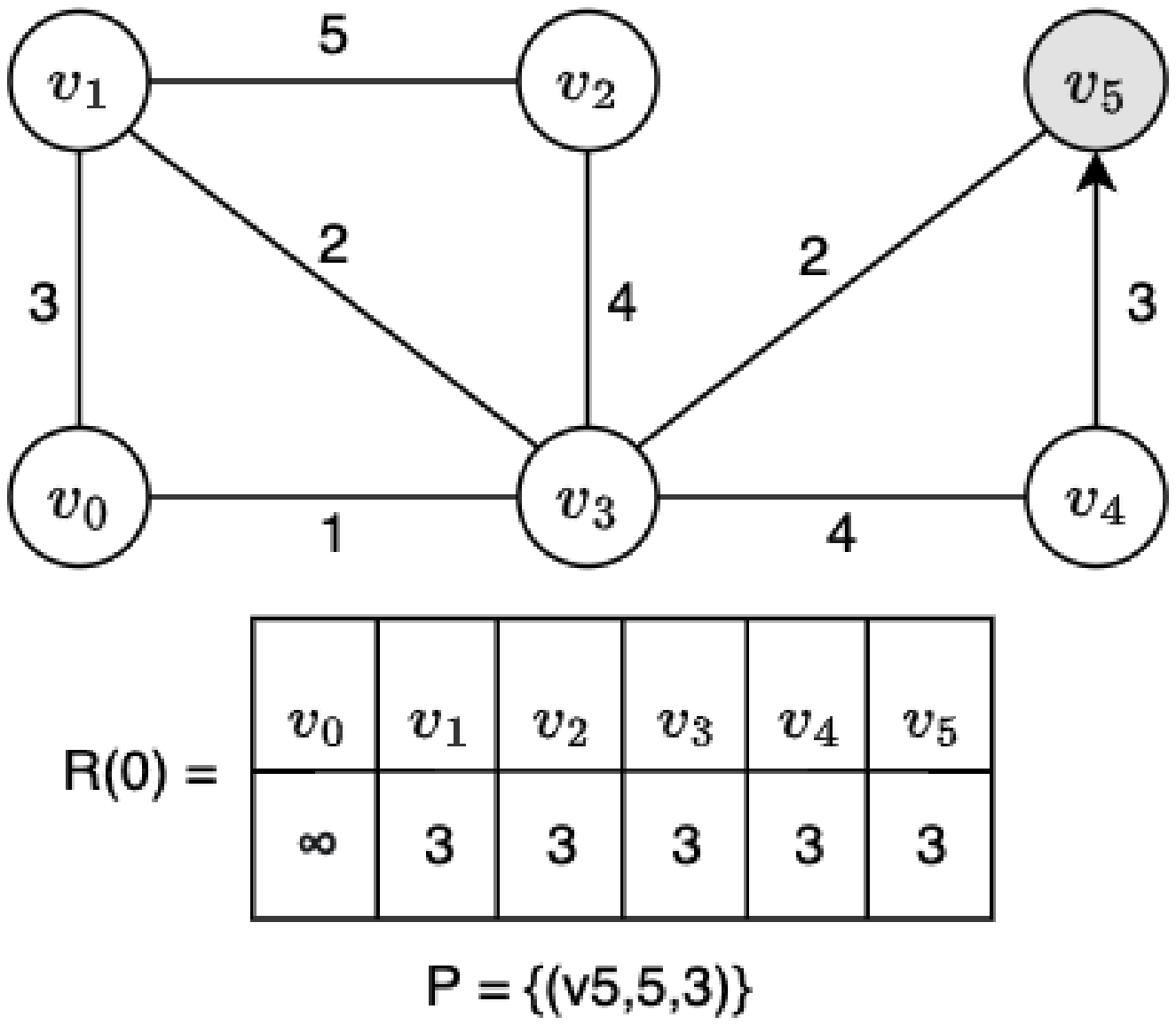}%
  \label{fig:bfs_e}
}\hfill
\subfloat[The $6^{th}$ iteration.\label{sfig:testa}]{%
  \includegraphics[width=.3\linewidth]{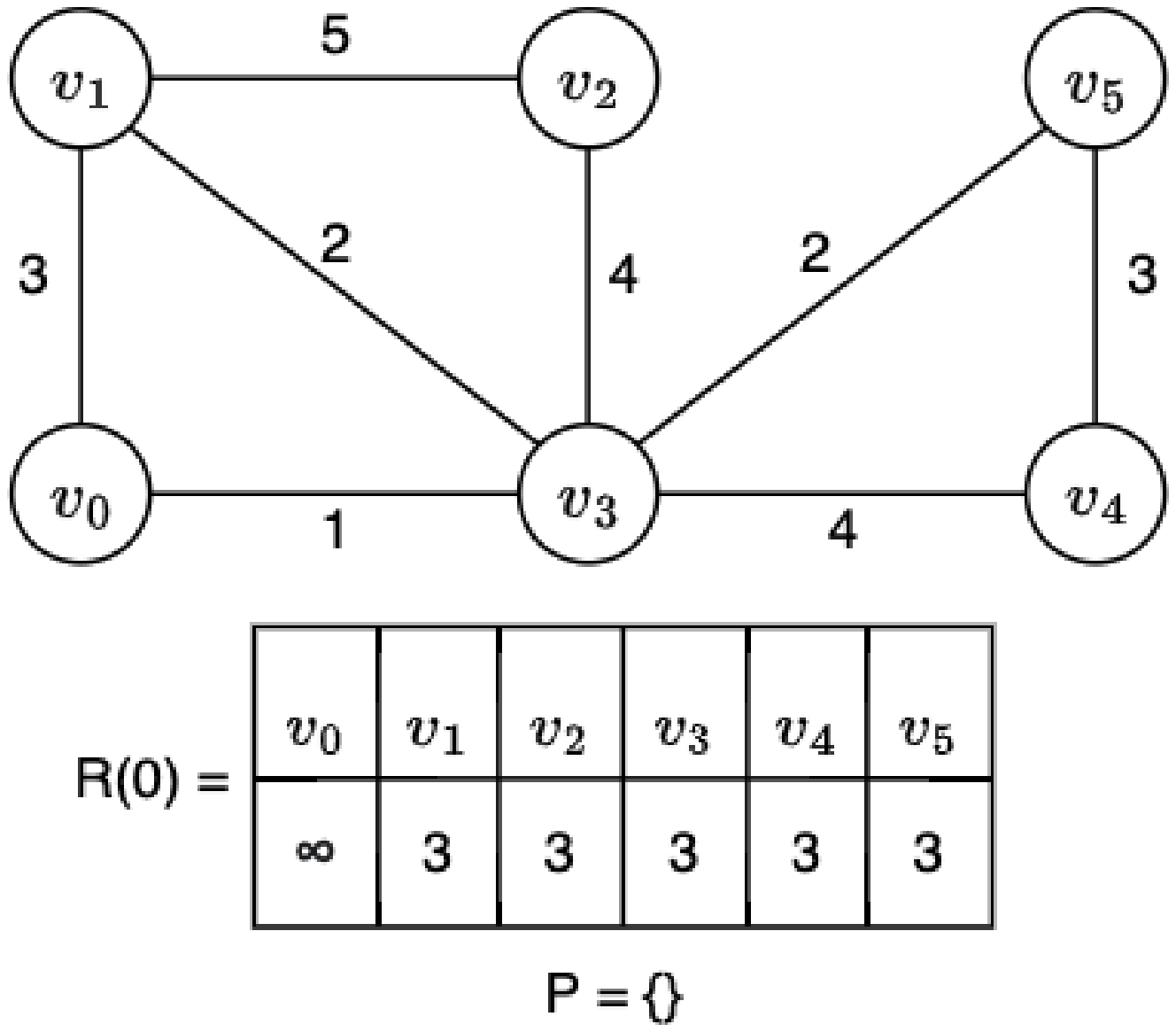}%
  \label{fig:bfs_f}
}
\caption{The constrained BFS process for $v_0$.}
\label{fig:BFS_example}
\end{figure*}

\subsection{Distance-Prioritized Search Order}
The whole index construction process consists of $|V|$ iterations' constrained BFS starting from different vertices. In each constrained BFS, it will explore at most $|V|$ vertices, but each vertex will be touched at most once. The order of these $|V|$ iterations' starting vertex is named \textit{Vertex Order}, while in $i$-th constrained BFS starting from vertex $v_i$, the order to explore remaining vertices is called the \textit{search Order} of $v_i$. These two types of orders are crucial for indexing time, indexing size, and query time in the 2-hop based index.



Below we introduce three properties that we aim to preserve for \ourindex. Then, a smart search order is proposed to guarantee these three properties at no additional cost, particularly the \textit{minimal} property.
\begin{itemize}
    \item \underline{\textit{Soundness}}.~If there are two index entries $(v,d_1, w_1) \in 
L(s)$ and $(v, d_2, w_2) \in L(t)$ with $ w_1 \leq w_2$ ($w_2 \leq w_1$), then there exists a quality constrained path from $s$ to $t$ with distance $d_1 + d_2$ that satisfies quality constraint $w \leq w_1$ ($w_2$).
    \item \underline{\textit{Completeness}}.~If there is a quality constrained shortest path $P_0$ from $s$ to $t$ with distance $d$, and satisfying quality constraint $w$ ($w = min(w(e) | e \in P_0 )$), then there exist either two label entries  $(v,d_1, w_1) \in L(s)$ and $(v, d_2, w_2) \in L(t)$. If $w = min(w_1, w_2)$, then $d_1 + d_2 = d$ and $w = w$, or single label entry such as $(s, d, w)$ $\in L(t)$ or $(t, d, w)$ $\in L(s)$.
    \item \underline{\textit{Minimal}}.~Intuitively, the minimal property indicates that any deletion of the existing label entries will cause incorrect results for some queries. This property is formulated as follows: For a vertex $u$, an entry $I = (v, d_1, w_1)$ is minimal if $I$ is not dominated by any other entries in $L(u)$; that is, there is no entry $I' = (v, d_2, w_2) \in L(v)$ s.t. $d_2 \leq d_1$, and $w_2 \geq w_1$. An index entry $I = (v, d_1, w_1)$ is \textit{necessary} if there does not exist a vertex $u_0$ s.t. $(u_0, d_0, w_0) \in$ $L(s)$ and $(u_0, d_0', w_0') \in$ $L(t)$ where $d_0 + d_0' \leq d_1$ and $min(w_0, w_0') \geq w_0$. Then, a \ourindex \ is \textit{minimal} if every index entry in it is both minimal and necessary.
\end{itemize}

To efficiently construct the \ourindex, the dominance relationships between edge qualities are exploited. Utilizing path domination, pruning is performed by traversing vertices in a certain order. To optimize the number of path traversals that are pruned throughout the index construction process, the following priority-based search orders are strictly adhered:

\begin{enumerate}
    \item \textit{Distance order}.~Computing the index entries with smaller distance $d$ first;
    \item \textit{Quality order}.~When tackling one specific $d$ value, explore the entries with the largest quality value $w$ first.
\end{enumerate}

Based on the above processing order, the WC-INDEX is constructed using BFS traversals from each vertex. Consider the BFS process from vertex $v \in V$. The maximum $w$ value of paths from $v$ to all other vertices are recorded. During each iteration of BFS expansion, it will be determined whether the visited vertices, say $u$, can be reached from $v$ by an existing path that dominates the current path, where an existing path is a path indicated by the current index entries. If the current path from $v$ to $u$ is dominated by an existing path, $u$ is pruned from the BFS process. Otherwise, the corresponding index entries are added into the WC-INDEX. Before moving onto the next iteration, all paths are processed from this iteration of expansion. Therefore, it is guaranteed that the index entries added in this iteration will not be dominated by any other entries.

\begin{algorithm}[thb]
	\SetAlgoVlined
	\SetFuncSty{textsf}
	\SetArgSty{textsf}
	\small
	\topcaption{WC-Index Construction} \label{alg:index-construction-distance}
	\Input {a graph $G$, and a vertex order $\mathcal{O}$;}
	\Output {the constructed 2-hop index $\mathcal{L}$}
	\State{$\mathcal{L}(v) \leftarrow \{(v,0,\infty)\}$ for all $v \in V(G)$;} \\
	\label{line:lv}
	\For{ $k = 1,2,\cdots, n$} { \label{naive-iterate-every-vertex}
	    \State{$v_k \leftarrow $ the $k$-th vertex in $\mathcal{O}$; } \\
	    \State{{$R(v) \leftarrow 0$ for all $v \in V(G)$;}} \\
	    \label{line:rv}
	     \State{$P \leftarrow$ an empty queues;} \\
	     \label{line:Pinit}
    \State{$P.push((v_k,0,\infty))$;} \\
    \While{$P \neq \emptyset$}{
    \label{dist-p-beg}
        \State{$vec \leftarrow \emptyset$;} \\
        \While{$P \neq \emptyset$}{
            \State{$(u, d, w) \leftarrow P.pop()$;} \label{dist-p-start}\\
            \State{\textbf{if} \textsc{Query}$(v_k,u, w, \mathcal{L}) \leq d$ \textbf{then continue};}\\
            \label{cons:pruned1}
            \State{\textbf{else} $\mathcal{L}_{k}(u) \leftarrow \mathcal{L}_{k}(u) \bigcup (v_k, d, w)$;} \\
            \label{cons:insert}
            \ForEach{$v_i \in N_G(u) {:O(v_i) > O(v_k)}$}{ \label{distance-for-start}
                \State{$w' \leftarrow \min(\delta(e=(u,v_i)), w)$;} \\
                \State{{\textbf{if} $w' \leq R(v_i)$ \textbf{then continue}};} \\
                \State{$vec \leftarrow vec \bigcup \{v_i\}$; $R(v_i) \leftarrow w'$;} \label{dist-p-end} \\
            }
        }
        \State{\textbf{foreach} $w \in vec$ \textbf{do} $P.push(v_i, d + 1, R(v_i))$;} \\
        \label{line:insert}
    }
    \label{line:bfs_end}
	}
	\State{\Return $\mathcal{L}$;}\\ 
\end{algorithm}

\stitle{Details}.~The algorithm for constructing \ourindex \ is shown in Algorithm \ref{alg:index-construction-distance}. Given graph $G$ and a vertex order $O$, this algorithm constructs the \ourindex \ $L$, which consists of entry sets $L(v)$ for every $v \in V$. Each entry set $L(v)$ is initialized as a set that contains only one entry, which corresponds to $v$ itself (Line~\ref{line:lv}). Then, BFS is executed for all $v_k \in V$ following the specified order. A vector $R(v)$ of size $|V|$ is used to record the current largest $w$ value of all paths from $v$ to all other vertices in the graph, with values set to 0 (Line~\ref{line:rv}). The maximum $w$ value from $v_k$ to $u$ is denoted as $w_{max}^{u}$. A queue $P$ is used to store tuples in the form of $(u,d,w)$, where $u$ is a vertex visited in the previous round of BFS, $d$ is the associated BFS path length, and $w$ is the minimum edge quality of that path. $P$ is initialized to contain a single element of $(v_k,0,\infty)$ (Line~\ref{line:Pinit}). The BFS process from $v_k$ is described in Line~\ref{dist-p-beg}-\ref{line:bfs_end} of Algorithm~\ref{alg:index-construction-distance}. 

During each iteration of BFS expansion, for each entry in queue $P$, a query is performed on the $w$-constrained path from $v_k$ to $u$ using the current index constructed so far (Line ~\ref{cons:pruned1}). This entry will be pruned if the result $w$-constrained distance from the query is smaller than the current BFS distance $d$. If not, the entry is appended to the index (Line~\ref{cons:insert}). Then, for each $v_i \in N(u)$, it will determine whether $v_i$ can be reached from $v$ by an alternative path with a greater $w$ value (Line~\ref{distance-for-start}-\ref{dist-p-end}). This is determined by comparing the current $w$ to $w_{max}^{v_i}$. If $w < w_{max}^{v_i}$, then $w$ is pruned from the BFS process. Otherwise, $v_i$ is added to a temporary set, and update $w_{max}^{v_i}$ with the value $w$. After all neighbors of $u$ have been processed, all temporary queue entries are pushed into $P$ to be processed in the next expansion iteration, with a distance of one step further from $v_k$ (Line~\ref{line:insert}). Thus, the algorithm ensures that for each $v_i$ only one path with the greatest $w$ will be considered in the next iteration. After all potential entries are popped from $P$, the process is repeated on this queue for the following round of BFS. The entire BFS from $v_k$ ends when $P$ is empty. The construction of \ourindex \ index finishes, after performing BFS for all $v \in V$.

\begin{example}
Figure~\ref{fig:BFS_example} illustrates how the Algorithm~\ref{alg:index-construction-distance} operates for the vertex $v_0$ in Figure~\ref{fig:running_example}. $R(v) \leftarrow 0$ for every $v \in V(G)$ and $R(v_0) = \infty$. Figure~\ref{fig:bfs_a} investigates the neighbors of $v_0$, i.e., $v_1$ and $v_3$. Then, $R(v_1) = 3$ and $R(v_3) = 1$. In addition, $P$ is updated to include the newly added vertices $v_1$ and $v_3$. Figure~\ref{fig:bfs_b} indicates that $v_2, v_3, v_4$, and $v_5$ will be explored. It is noted that $v_3$ is updated into $P$ again since in the round, $R(v_3)$ is updated with a larger value, i.e., $2$. Likewise, in Figure~\ref{fig:bfs_c}, $v_3, v_4$ and $v_5$ are inserted into $P$ due to their updated $R$ values. In Figure~\ref{fig:bfs_d}, only $v_4$ is inserted into $P$ and $R(v_4)$ is updated with $3$ since a path $v_0 \rightarrow v_1 \rightarrow v_2 \rightarrow v_3 \rightarrow v_4$ is found. In this path, the minimal quality is $e(v_0, v_1) = 3$. Figure~\ref{fig:bfs_e} depicts the updates for $v_5$ with only $v_5$ being inserted into $P$ and $R(v_5)$ being updated with value $3$. This is the result of the newly found path $v_0 \rightarrow v_1 \rightarrow v_2 \rightarrow v_3 \rightarrow v_4 \rightarrow v_5$. Figure~\ref{fig:bfs_f} illustrates the last iteration. The constrained BFS for $v_0$ terminates at this iteration since there is no update for any vertex. For every triple inserted into $P$, the corresponding label entry is inserted into $L(v_0)$.
\end{example}

\begin{lemma}
\label{lemma:dom}
In \refalg{index-construction-distance}, for each candidate index entry popped by the queue (Line~\ref{dist-p-start}), it cannot be dominated by all the candidate index entries popped by the queues afterwards.
\end{lemma}
\begin{proof}
Since for each iteration of BFS expansion, all entries in queue $P$ are popped, and then new entries are added back to this queue $P$. Then, this queue can only contain entries with the same $d$ at a given moment. For vertex $u$ in a certain BFS iteration, if $u$ already exists in the temporary set, which indicates there is an existing path from $v_k$ to $u$ with $w_{max}^{u}$. If the current path induces an entry with $w > w_{max}^{u}$, $w_{max}^{u}$ is updated to be $w$. Otherwise, nothing happens. Consequently, $w$ will only exist once in the temporary set, and will only be pushed into queue $P$ once, with $w = w_{max}^{u}$, which is the maximum $w$ value at distance $d$. Any future entries in $P$ regarding $u$ will include a larger $d$. Therefore, popped entries will never be dominated by future entries in the queues.
\end{proof}

\noindent\textbf{{Correctness of Algorithm}}.~Then, the correctness of the algorithm is proved by its \textit{Soundness} and \textit{Completeness}. Additionally, the \textit{Minimal} properties is proved. 

\begin{theorem}
\label{theo:property}
Algorithm~\ref{alg:index-construction-distance} can construct a \textit{Sound}, \textit{Complete}, and \textit{Minimal} index for \ourproblem \ problem.
\end{theorem}
\begin{proof}
First, the \textit{Soundness} and \textit{Completeness} are demonstrated. These two characteristics are equivalent to the correctness of Algorithm~\ref{alg:index-construction-distance}.

\noindent\underline{\textit{Soundness}}.~It is proved by contradiction. Assume there are two index entries $(v,d_1, w_1) \in L(s)$ and $(v, d_2, w_2) \in L(t)$ with $ w_1 \leq w_2$ ($w_2 \leq w_1$), and there does not exist a quality constrained path from $s$ to $t$ with distance $d_1 + d_2$ and satisfy quality constraint $w \leq w_1$ ($w_2$). According to the index construction process, there are two quality constrained shortest paths. The first is $s \rightsquigarrow v$ with distance $d_1$ and quality constraint $w_1$, whereas the second is $v \rightsquigarrow t$ with distance $d_2$ and quality constraint $w_1$. Therefore, it can be combined to produce a new path $P_{new}$. Note that \textit{soundness} simply requires a quality constraint path; it does not have to be shortest.

\noindent\underline{\textit{Completeness}}.~Similarly, the \textit{Completeness} is demonstrated by contraction. Assume that there is a quality constrained shortest path from $s$ to $t$ with distance $d$, satisfying quality constraint $w$, then there does not exist either two label entries  $(v,d_1, w_1) \in L(s)$ or $(v, d_2, w_2) \in L(t)$. If $w = min(w_1, w_2)$, then $d_1 + d_2 = d$, nor one label entry like $(s, d, w)$ $\in L(t)$ or $(t, d, w)$ $\in L(s)$. Assume $s$ is explored before $t$\footnote{The proof process is similar if $t$ is the earlier one.} and $s$ is the first vertex that leads to such incorrectness, and $s ,t$ is the first vertex pair to lead the incorrectness. This indicates that the \textit{Completeness} of all the previously explored vertices is maintained. Consequently, according to Algorithm~\ref{alg:index-construction-distance} Line~\ref{cons:pruned1}, if the $Query(s, t, w)$ is pruned, then it indicates that there exist two label entries  $(v,d_1, w_1) \in L(s)$ and $(v, d_2, w_2) \in L(t)$. If $w = min(w_1, w_2)$, then $d_1 + d_2 = d$. Otherwise, the label entry $(s, d, w)$ is inserted into $L(t)$ in accordance with Algorithm~\ref{alg:index-construction-distance} Line~\ref{cons:insert}.

\noindent\underline{\textit{Minimal}}.~According to Algorithm~\ref{alg:index-construction-distance} Line~\ref{cons:pruned1}, a newly added label entry is \textit{Minimal} when it is inserted into the index. Therefore, it is only necessary to prove it will not be dominated in the label entries that are inserted after it. Due to the \textit{distance order}, \textit{quality order}, and Definition~\ref{def:dom}, this property is automatically maintained. 
\end{proof}

\begin{theorem}
The index constructed by Algorithm~\ref{alg:index-construction-distance} is capable of producing correct results.
\end{theorem}
\begin{proof}
Theorem~\ref{theo:property} proves the \textit{Soundness} and \textit{Completeness} of constructed index. Thus, its correctness is immediately proved.
\end{proof}

\stitle{Complexity Analysis}.~The while loop dominates the time complexity of indexing from the vertex $v_k$. Let $I(v)$ denote all the index entries associated with $v$ and let $\zeta = \max\limits_{v\in V(G)}|I(v)|$. Let $d_{max}$ denote the maximum vertex degree in the graph. Observe that in \refalg{index-construction-distance}, Lines~\ref{distance-for-start}-\ref{dist-p-end} are executed at most $\zeta$ times, hence the size of the priority queue cannot exceed $\zeta d_{max}$. For each index entry in the queue, a query operation is performed to determine whether it can be covered by the existing index entries, and the query time is bounded by $O(\zeta)$. As a result, the time complexity of \refalg{index-construction-distance} is $O(n \cdot \zeta \cdot d_{max} (\log \zeta \cdot d_{max} + \zeta))$.

The size of the index is bounded by $O(\sum_{u \in V(G)}\sum_{v \in V(G)_{\leq u}}min(D, |w|)).$ 


\subsection{Query-Efficient Implementation}
Since the $Query$ function is commonly utilized during the index construction and query stages, it is a vital component that influences three aspects of the index: indexing time, index size, and query time. This subsection investigates how to efficiently implement the $Query$ function by utilizing the problem's property.

Given a query$(s, t, d, w)$, a basic operation is to determine whether there are two label entries $(u_1, d_1, w_1)$ and $(u_2, d_2, w_2)$ with $u = u_1 = u_2$, $d_1 + d_2 \leq d$, $w_1 \geq w$, and $w_2 \geq w$. 

\begin{algorithm}[thb]
	\SetAlgoVlined
	\SetFuncSty{textsf}
	\SetArgSty{textsf}
	\small
	\topcaption{$Query$} \label{alg:query_naive}
	\Input{any two vertices $s,t \in V$, constraint $w$, and current distance $d$;}
	\Output{a boolean value indicating if a path is found}
	\For{$\forall I_j$ $\in$ $L[t]$} {
	\label{naive:for}
	    \If{$I_j.vertex > s$ or $I_j.quality < w$}{\textbf{continue};}
	    \label{naive:prune1} 
	    \State{$v = I_j.vertex$;}\\
	    \If{$L[s][v] = \emptyset$} {\textbf{continue};}
	    \For{$\forall I_i$ $\in$ $L[s][v]$} {
	        \If{$I_i.quality \geq w$} {
	            \If{$I_i.dist + I_j.dist <= d$} {
	                \State{\Return $True$;}\\ 
	            }
	        }
	    }
	}
	\State{\Return $False$;}\\ 
\end{algorithm}

\noindent\textbf{Na\"ive Implementation}.~For simplicity, $L[u]$ denotes all the label entries of vertex $u$, and $L[u][v]$ denotes all the label entries as $(v, d_v, w_v)$ in $L[u]$.
The na\"ive query function is represented by Algorithm~\ref{alg:query_naive}. Line~\ref{naive:for} traverses every label entry in the $L[t]$. Assume a label entry is $I_j$, Line~\ref{naive:prune1} prunes it if its vertex order is larger than $s$ or $quality$ is less than the quality constraint. Otherwise, entries of $L[s][v]$ are explored, where $v$ is the vertex of $I_j$, and validate whether there are two valid label entries $I_j$ and $I_i$ to return a $true$ result. The time complexity of this implementation is $O(|L(s)| + |L(t)| + \sum_{v \in L[t].vertex}{ | L[t][v] | \times | L[s][v] |})$.

The following theorem helps speed up this procedure.
\begin{theorem}
\label{theo:increase}
For two label entries $(u_0, d_0, w_0)$ and $(u_0, d_1, w_1)$ in $L(v)$, if $d_0 > d_1$, then $w_0 > w_1$, and vice versa.
\end{theorem}
\begin{proof}
This theorem is proved by contradiction. Assume that there are two label entries $(u_0, d_0, w_0)$, and $(u_1, d_1, w_1)$, s.t. $d_0 > d_1$ and $w_0 \leq w_1$. According to Lemma~\ref{lemma:dom}, $(u_0, d_0, w_0)$ will be eliminated since it is dominated by $(u_1, d_1, w_1)$, which results in a contradiction. Likewise, a similar contradiction exists when $w_0 > w_1$ with $d_0 \leq d_1$.
\end{proof}

\stitle{Querying}.~During the BFS process for one vertex $v_k$, it is noted that all queries are issued with one end-point as $v_k$. Therefore, an array $T$ of size $|V|$ is initialized with the existing index entry of that $v_k$ before the BFS begins. To evaluate $QUERY(u,v_k)$, the new querying algorithm needs $O(|L(u)|)$ time rather than $O(|L(u)|) + O(|L(v_k)|)$ for looping through two entry lists.

Based on Theorem~\ref{theo:increase}, the index entries $(d_i,w_i), i = 1,2,...,|L(v)|$ for vertex $v \in V$ must be in increasing order in terms of $d$ and $w$. If $j > i$, then both $d_j > d_i$ and $w_j > w_i$. Instead of iterating through $T$, binary search could be utilized to locate elements. Then, in such an implication, the time complexity is $O(|L(s)| + |L(t)| + \sum_{v \in L[t].vertex}{ | L[t][v] | \times log| L[s][v] |})$.

\noindent\textbf{Query-Efficient Implementation}.~Based on Theorem~\ref{theo:increase}, the time complexity can be further reduced to $O(|L(s)| + |L(t)|)$. The idea is explained as follows: Since the index entries $(d_i, w_i), i = 1,2,..., |L(v)|$ for vertex $v \in V$ must be in ascending order in terms of $d$ and $w$, if finding the first index entry $(u_i, w_i, d_i)$ whose $w_i \geq w$, $d_i$ is minimal for $(u_i, \cdot, \cdot)$. Thus, for every $v \in L(s)$ or $L(t)$, only one label entry is required. Then, a na\"ive scanning could be conducted to answer the queries. The time complexity is $O(|L(s)| + |L(t)| + \sum_{v \in L[t].vertex}{ (log| L[t][v] | + log| L[s][v] |)})$. Since $\sum_{v \in L[t].vertex}{ (log| L[t][v] | + log| L[s][v] | )}$ $\leq$ $|L(s)| + |L(t)|$, the final time complexity is $O(|L(s)| + |L(t)|)$.

\noindent\textbf{Details}.~The details of the Query-Efficient Implementation is illustrated in Algorithm~\ref{alg:query_efficient}. Line~\ref{eff:for} traverses every vertex $v \in L(t)$. Line~\ref{eff:prune} prunes if $L[s][v] = \emptyset$. If not empty, a modified binary search is utilized to locate the first label entry with $w_i \geq w$ in $L[t][v]$  and $w_j \geq w$ in $L[s][v]$, respectively. $true$ is immediately returned if $d_1 + d_2 \leq d$ in Line~\ref{eff:true}. Otherwise, the query answer is $false$ in Line~\ref{eff:false};

\begin{algorithm}[thb]
	\SetAlgoVlined
	\SetFuncSty{textsf}
	\SetArgSty{textsf}
	\small
	\topcaption{$Query^+$} \label{alg:query_efficient}
	\Input{any two vertices $s,t \in V$, constraint $w$, and current distance $d$;}
	\Output{a boolean value indicating if a path is found}
	\For{$\forall$ vertex $v$ $\in$ $L[t]$}{
	\label{eff:for}
	    \If{$L[s][v] = \emptyset$} {\textbf{continue};}
	    \label{eff:prune}
	    \State{Find $I_i \in L[t][v]$ which is the first label entry with $w_i \geq w$;}\\
	    \label{eff:first1}
	    \State{Find $I_j \in L[s][v]$ which is the first label entry with $w_j \geq w$;}\\
	    \label{eff:first2}
	    \If{$d_i + d_j \leq d$}
	    {\State{\Return $True$;}}
	    \label{eff:true}
	}
	\State{\Return $False$;}\\ 
	\label{eff:false}
\end{algorithm}

\stitle{Efficient Initialization}.~An important aspect is to avoid $O(n)$ time initialization for data structures during each round of BFS. This may develop into a bottleneck. A solution is to set updated values in the array, without recreating the whole array. This can be accomplished by recording which vertices have been processed during the process, and only update them.

\stitle{Further Pruning}.~Whenever a path is found during the query process of index construction, the algorithm records the result for the current quality of that vertex pair. If a potential query in the same BFS round has the same vertex pair and a quality not greater than the recorded quality, the query process can be skipped since its result is recorded.


\subsection{Vertex Ordering Strategies}
Vertex ordering is one of the vital orders that significantly affect indexing time, index size, and querying time. This subsection investigates a hybrid vertex ordering based on some observations.

\begin{observation}
\label{obs:degree}
The degree ordering is shown to have better performance than other orderings~\cite{akiba2013fast} for the shortest path distance problem in the scale-free network, e.g., social networks. Notwithstanding, for the road network ``Indochina"\footnote{\url{http://law.di.unimi.it}}, the tree decomposition based ordering has much better performance. 
\end{observation}

\begin{observation}
\label{obs:tree}
It is shown in ~\cite{ouyang2018hierarchy} that Vertex Hierarchy via Tree Decomposition technique is appropriate for the road network for distance query.
\end{observation}

To use the Observation~\ref{obs:tree}, it first introduces the degree-based ordering as well as the Vertex Hierarchy through Tree Decomposition.

\noindent \underline{\textit{Degree-Based Scheme}}.~A vertex with a higher degree is likely to cover more shortest paths. In summary, in degree-based ordering, vertices are sorted in non-ascending order of degree. This scheme leads to the state-of-the-art canonical hub labeling for shortest distance queries.

\noindent\underline{\textit{Tree Decomposition Ordering}}. Tree decomposition is a technique for mapping a graph to a tree in order to accelerate the resolution of certain computational problems in graphs ~\cite{halin1976s, robertson1984graph}. Numerous algorithmic problems, such as maximum independent set and Hamiltonian circuits that are NP-complete for arbitrary graphs, can be solved efficiently by dynamic programming for graphs of finite treewidth, employing the tree-decompositions of these graphs. A summary of Bodlaender's introduction can be found in~\cite{bodlaender1994tourist}. The tree decomposition provides a natural hierarchy to vertices. In this paper, tree decomposition is utilized to establish the vertex hierarchy, and demonstrate that the hierarchy is effective in resolving quality constrained distance queries in networks. A tree decomposition of a graph $G(V, E)$ is defined as follows~\cite{bodlaender1994tourist}:
\begin{definition}[Tree Decomposition] A tree decomposition of a graph $G(V,E)$, denoted by $T_G$, is a rooted tree in which each node $X \in V(T_G)$ is a subset of $V(G)$ (i.e., $X \subset V(G)$) with the following three conditions:
\begin{itemize}
    \item $\bigcup_{X \in V(T_G)} X = V$;
    \item For every $(u,v) \in E(G)$, there exists $X \in V(T_G)$ s.t. $u \in X$ and $v \in X$.
    \item For every $v \in V(G)$ the set $\{ X | v \in X \}$ forms a connected subtree of $T_G$.
\end{itemize}
\end{definition}

Based on Observations~\ref{obs:degree} and~\ref{obs:tree}, it simply employed vertex ordering of the Vertex Hierarchy via Tree Decomposition in~\cite{ouyang2018hierarchy} and developed a fast approach to obtain this ordering as opposed to constructing their whole index for the \ourproblem \  problem.

The computation of the treewidth of a graph has been shown to be NP-Complete~\cite{arnborg1987complexity}. One of the most effective heuristics Tree decomposition is based on minimum degree elimination.

\noindent \textbf{Minimum Degree Elimination (MDE)-based Tree Decomposition}.~Minimum Degree Elimination~\cite{berry2003minimum} based tree decomposition removes recursively the vertex $v$ in $G$ with the minimum degree and then adds $v$'s neighbors' clique back to $G$. A
 bag of the tree decomposition is comprised of each node $v$ and its neighbors on the transient graph right before the deletion of $v$.
\begin{definition}[Minimum Degree Elimination]
Generate $n$ bags of nodes $\{ B_1, B_2,..., B_n\}$ and a sequence of nodes $\{ v_1, v_2,..., v_n\}$ in $n$ rounds with the starting graph $G_0 = G$. In the $i-th$ round, $i$ takes value from $1$ to $n$:
\begin{itemize}
    \item $v_i$ : the node with the lowest degree (or any one of these nodes if there is a tie situation) in $G_{i-1}$.
    \item $N_i$ : the neighbor set of $v_i$ in $G_{i-1}$.
    \item $B_i$ : $\{ v_i\} \cup N_i$.
    \item $G_i$ : a graph that eliminates $v_i$ from $G_{i-1}$ and then adds clique($N_i$), that is $V(G_i)$ $= V(G_{i-1} \backslash \{ v_i \}$, and $E(G_i) = E(G_{i-1})$ $\cup$ $E[clique(N_i)]$ $\backslash \{ N_i \} \times N_i$. 
\end{itemize}
\end{definition}

\noindent\underline{\textit{Hybrid Vertex Ordering}}.~Therefore, this paper proposes a hybrid vertex ordering that compromises between the computational efficiency of degree vertex order and the index size effectiveness of the tree decomposition order as follows:

\begin{itemize}
    \item \textit{Classification}.~All vertices are classified into two categories: core part and periphery. To achieve this, a degree threshold $\delta$ is specified. If a vertex $v$'s degree is above this threshold, it is classified into the core-part. Otherwise, it is classified into the periphery.
    \item \textit{Core-Part}.~Regarding the core-part vertices, it is observed that the computation cost can be quite high if the tree decomposition method is used. Therefore, all these vertices are ordered according to their degree.
    \item \textit{Periphery}.~The vertices in periphery are ranked according to tree decomposition order.
    \item \textit{Combinations}.~Then, these two types of vertices are combined to produce a hybrid vertex order.
\end{itemize}
\section{Variants and extensions}
\label{sect:ext}

\noindent \textbf{Quality Constrained Shortest Path}.~Similar to~\cite{akiba2013fast}, to locate the exact shortest path rather than the distance, the modified algorithm records sets of quads instead of triples of labels. Let $L(v)$ be a set of quads of $(u, d_u, w_u, p_{uv})$, where $p_{uv} \in V$ is the last edge visited before inserting this label entry in the index construction search process of Algorithm~\ref{alg:index-construction-distance} starting from $u$. It can restore the shortest path between $v$ and $u$ by ascending the last edge from $v$ to the parents.

\noindent \textbf{Directed and Weighted Graphs}.~To modify \ours \ to a directed graph, the only modification required is to conduct a constrained constrained BFS from two directions for each vertex. In addition,  $L_{in}$ and $L_{out}$ are required to hold the index data for in-coming edges and out-coming edges, respectively. It is necessary to traverse all the index entries in $L_{out}(s)$ and $L_{in}(t)$ for a $query(s, t, w)$. In cases where the length of an edge is not $1$ (e.g., weighted graph), we can convert the constrained BFS to a constrained Dijkstra.




\section{Experimental Evaluations}
\label{sect:exp}

\vspace{1mm}
\noindent\textbf{Datasets}.~Tables~\ref{tab:summary_of_road_networks} and~\ref{tab:summary_of_social_networks} provide the statistics of real graphs used in the experiments. 14 publicly available datasets are used. These datasets can be downloaded from either KONECT\cite{kunegis2013konect}\footnote{\url{http://konect.uni-koblenz.de}} or SNAP~\cite{leskovec2011stanford}\footnote{\url{https://snap.stanford.edu}}. Directed graphs were converted to undirected ones in our testings. \rev{For labeled graphs such as Movielens, $|w|$ is directly taken from the original data-set. For other non-labeled graphs, we randomly generate those weights.} For query performance evaluation, 10,000 random queries were employed and the average time is reported.

\vspace{1mm}
\noindent\textbf{Settings}.~In experiments, all programs were implemented in standard c++11 and compiled with g++4.8.5.
\\
All experiments were performed on a machine with 20X Intel Xeon 2.3GHz and 385GB main memory running Linux(Red Hat Linux 7.3 64 bit).

\begin{table}[thb]
    \centering
    \small
    \topcaption{Summary of Road Networks} \label{tab:summary_of_road_networks}
    \scalebox{0.96}{
    \begin{tabular}{|c||c|c|c|c|c|} \hline
        \textbf{Name} & \textbf{Dataset} & $|V(G)|$ & $|E(G)|$ \\ \hline
        NY &New York City& 264,346 & 733,846 \\ \hline
        FLA & Florida & 1,070,376 & 2,712,798 \\ \hline
        CAL & California and Nevada & 1,890,815 & 4,657,742 \\ \hline
        E & Eastern USA & 3,598,623 & 8,778,114 \\ \hline
        W & Western USA & 6,262,104 & 15,248,146 \\ \hline
        CTR & Central USA & 14,081,816 & 34,292,496 \\ \hline
        USA & Full USA & 23,947,347 & 58,333,344 \\ \hline
    \end{tabular}}
\end{table}

\begin{table}[thb]
    \centering
    \small
    \topcaption{Summary of Social Networks} \label{tab:summary_of_social_networks}
    \scalebox{0.96}{
    \begin{tabular}{|c||c|c|c|c|} \hline
        \textbf{Name} & \textbf{Dataset} & $|V(G)|$ & $|E(G)|$ & $|w|$ \\ \hline
        MV-10 & Movielens-10m & 80,555 & 10,000,054 & 5 \\ \hline
        EU & eu-2005 & 862,664 & 16,138,468 & 3 \\ \hline
        ES & eswiki-2013 & 970,331 & 21,184,931 & 3 \\ \hline
        MV-25 & Movielens-25m & 221,588 & 25,000,095 & 5 \\ \hline
        FR & frwiki & 1,350,986 & 31,037,302 & 3 \\ \hline
        UK & uk-2007 & 1,000,000 & 37,061,970 & 3 \\ \hline
        SO-Y& Stackoverflow (year) & 2,601,977 &  28,183,518& 9 \\ \hline
    \end{tabular}}
\end{table}

\begin{figure}[htbp]
    \centering
    {\includegraphics[width=0.8\linewidth]{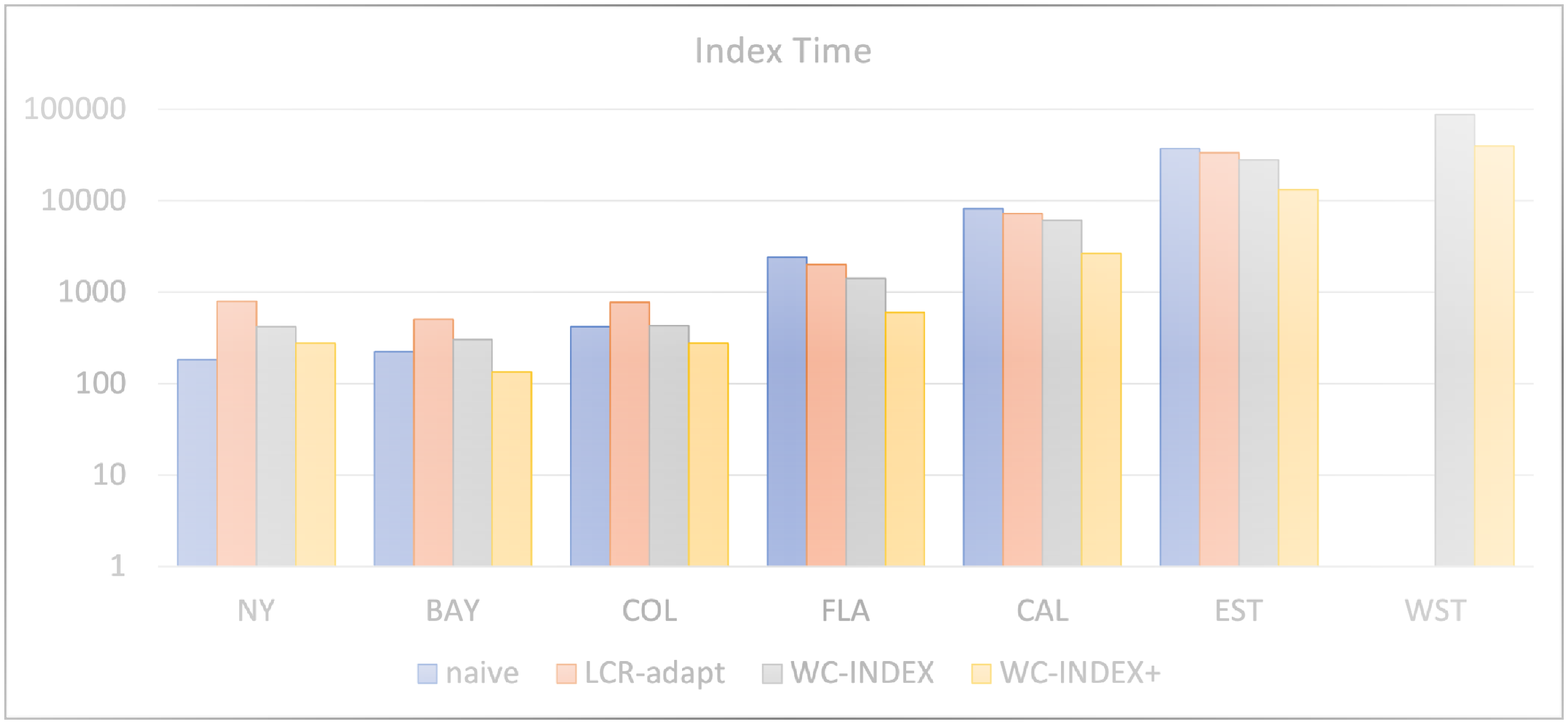}}
        \vspace{-1mm}
    \caption{Indexing Time (s) for baseline, \ours, and \ourp.}
    
    \label{fig:index_time}
    \vspace{-1mm}

\end{figure}

\begin{figure}[htbp]
    \centering
    {\includegraphics[width=0.8\linewidth]{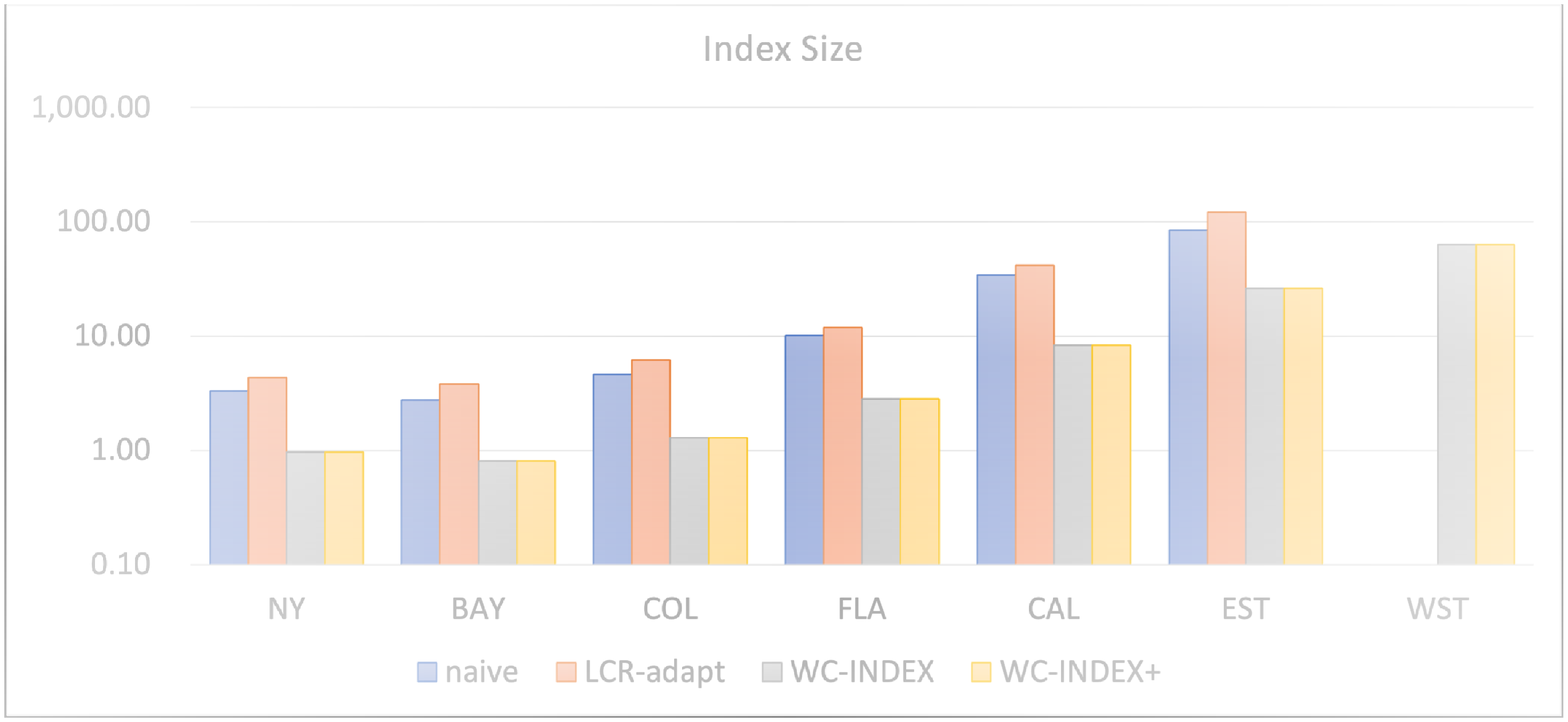}}
        \vspace{-1mm}

    \caption{Indexing Size (GB) for baseline, \ours, and \ourp.}
    \label{fig:index_size}
        \vspace{-1mm}

\end{figure}


\begin{figure*}[htbp]
    \centering
    {\includegraphics[width=1\linewidth]{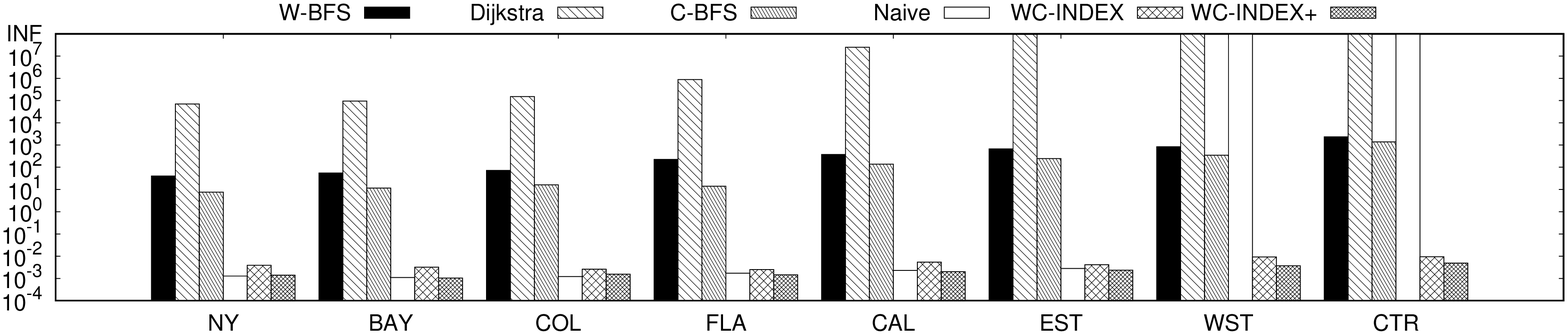}}
        \vspace{-3mm}

    \caption{Querying time (ms) for baselines, \ours, and \ourp.}
    \label{fig:query_time}
        \vspace{-3mm}

\end{figure*}

\vspace{1mm}
\noindent\textbf{Algorithms}~We compare our techniques with the following baseline solutions.
\begin{itemize}
    \item \textbf{W-BFS}.~The original graph is partitioned into $|w|$ parts, and then conduct BFS.
    \item \textbf{Dijkstra}.~After the partitioning of the original graph into $|w|$ parts, Dijkstra is conducted. 
    \item \textbf{C-BFS}.~It conducts Constrained BFS on the original graph, with the valid edges explored.
    \item \textbf{Na\"ive}.~The na\"ive 2-hop labeling method introduced in Section~\ref{sect:baselines}.
    \item \rev{\textbf{LCR-adapt}.~We modify the state-of-the art \underline{L}abel \underline{C}onstrained \underline{R}eachability algorithm to our problem.}
    \item \textbf{\ours}.~The basic algorithm for the quality-constrained shortest path problem.
e    \item \textbf{\ourp}.~The advanced algorithm with the query-efficient and hybrid order techniques.
\end{itemize}


\noindent \textbf{Exp 1: Indexing Time for Road Networks}.~Figure~\ref{fig:index_time} illustrates the indexing time for Na\"ive 2-hop labeling index, \ours, and \ourp. What stands out in these figures is that \ourp \ is the fastest method to construct the index among these three algorithms. For instance, for CTR, only \ourp \ can construct the 2-hop index. As for Na\"ive and \ours, \ours \ is slower than Na\"ive in small datasets, e.g., NY, BAY, COL, EST. Notwithstanding, \ours \ is much faster than Na\"ive for large datasets, e.g., WST and CTR. \rev{We observed that for smaller graphs, the construction overhead of WC-INDEX dominates the index construction time. As a result, WC-INDEX builds up slower than the baseline index. On the other hand, na\"ive index simply filter the graph based on every possible weight and construct simple 2-hop indexes for each filter graph. When the graphs are small, this can be done relatively quickly compared to WC-INDEX. However, as the size of the graphs gets large, building indexes for every separate filtered sub-graph is costing much more time, and is eventually outperformed by WC-INDEX, which only constructs one index.}

\noindent \textbf{Exp 2: Indexing size for Road Networks}. Figure~\ref{fig:index_size} depicts the index size for Na\"ive 2-hop labeling index, \ours, and \ourp. What is striking in this figure is that \ours \ and \ourp \ could achieve the same index size. The reason is that they use the same vertex ordering, and the Query-Efficient technique can only speed up the construction process, but does not have any impact on the index size. As for Na\"ive, its index size is the largest among these three in all datasets. \rev{Table \ref{tab:size_of_road_networks} summarizes the memory usage of storing the road networks.}

\begin{table}[thb]
    \centering
    \small
    \topcaption{Size of Road Networks} \label{tab:size_of_road_networks}
    \scalebox{0.96}{
    \begin{tabular}{|c||c|c|} \hline
        \textbf{Name} & \textbf{Dataset} & Size (GB) \\ \hline
        NY &New York City& 0.006 \\ \hline
        FLA & Florida & 0.025 \\ \hline
        CAL & California and Nevada & 0.043 \\ \hline
        E & Eastern USA & 0.082\\ \hline
        W & Western USA & 0.142\\ \hline
        CTR & Central USA & 0.319\\ \hline
        USA & Full USA & 0.54\\ \hline
    \end{tabular}}
\end{table}





\noindent \textbf{Exp 3: Query Time for Road Networks}. Figure~\ref{fig:query_time} demonstrates the query time for W-BFS, Dijkstra, C-BFS, Na\"ive, \ours, and \ourp. An interesting obervation is that Dijkstra is the slowest among all the algorithms. It is evident from Figure~\ref{fig:query_time} that W-BFS and C-BFS have comparable query efficiency. C-BFS is more efficient than W-BFS in terms of query time. These two BFS-based online algorithms can commit on all the datasets. \rev{The reason that Dijsktra is slower than BFS is that it reserved the distance priority queue and a distance vector $D[v]$ to store all the distance information to the start vertex $s$. Additional, once a new distance from $s$ to $v$ is found, it would be compared with $d[v]$. $d[v]$ would be updated if a shorter one is found. With these additional operations, the W-BFS would run faster than Dijskstra, but Dijsktra could directly cope with the case where edge distance is not $1$. Nevertheless, W-BFS could not directly address such a case.}
The query time for the index-based technique is substantially smaller than the online search based method. On average, 4-5 orders of magnitudes speedup can be achieved. Nevertheless, the Na\"ive 2-hop labeling index can not be constructed for CTR and WST, hence the query time for these two datasets is set as $INF$. As with \ours \  and \ourp, they can be constructed in all datasets with a feasible index size, indexing time, and query time in microseconds. \rev{For very large road networks such as WST and CTR, the na\"ive indexing cannot constructed due to memory constraint, since the method builds separate indices for each $w$. As a result, the query time cannot be tested and thus listed as infinity.}

\begin{figure*}[htbp]
    \centering
    {\includegraphics[width=0.8\linewidth]{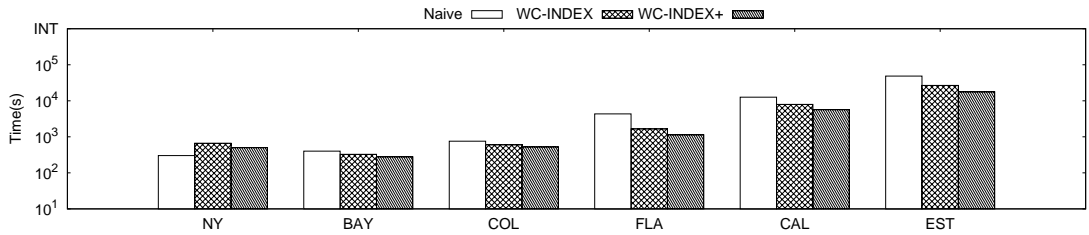}}
        \vspace{-3mm}

    \caption{Indexing time (s) for baseline, \ours, and \ourp, when $|W| = 20$.}
    \label{fig:index_time20}
        \vspace{-3mm}

\end{figure*}

\begin{figure*}[htbp]
    \centering
    {\includegraphics[width=0.8\linewidth]{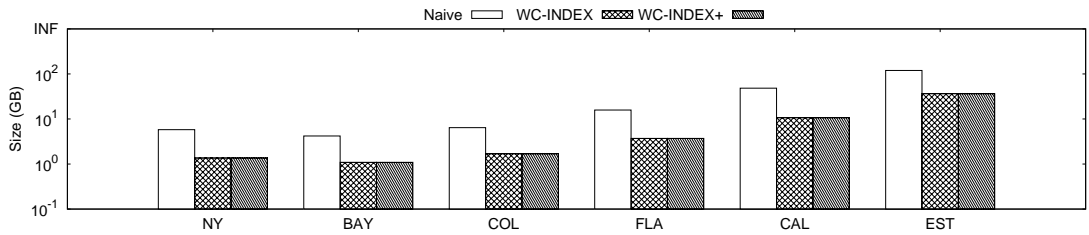}}
        \vspace{-3mm}

    \caption{Indexing size (GB) for baseline, \ours, and \ourp, when $|W| = 20$.}
    \label{fig:index_size_20}
        \vspace{-3mm}

\end{figure*}


\noindent \textbf{Exp 4: Large $|w|$}.~Exp 4 investigates the indexing time and indexing size for the number of different constraint values $|w| = 20$. Figures~\ref{fig:index_time20} and~\ref{fig:index_size_20} reports the findings. The results are similar to that in Exp 1 and 2. Regarding indexing time, Figure~\ref{fig:index_time20} reveals that \ourp \ is the fastest method among these three to construct the index. Regarding Na\"ive and \ours, \ours \ is slower than Na\"ive across all datasets evaluated, i.e., NY, BAY, COL, EST. As for indexing size, what is striking in Figure~\ref{fig:index_size_20} is that \ours \ and \ourp \ can achieve the same index size. The reason for this is because they both employ the same vertex ordering, and the Query-Efficient technique can only speed up the construction process, without affecting on the index size. As for Na\"ive, its index size is the largest among these three in all datasets.

\begin{figure*}[htbp]
    \centering
    {\includegraphics[width=0.8\linewidth]{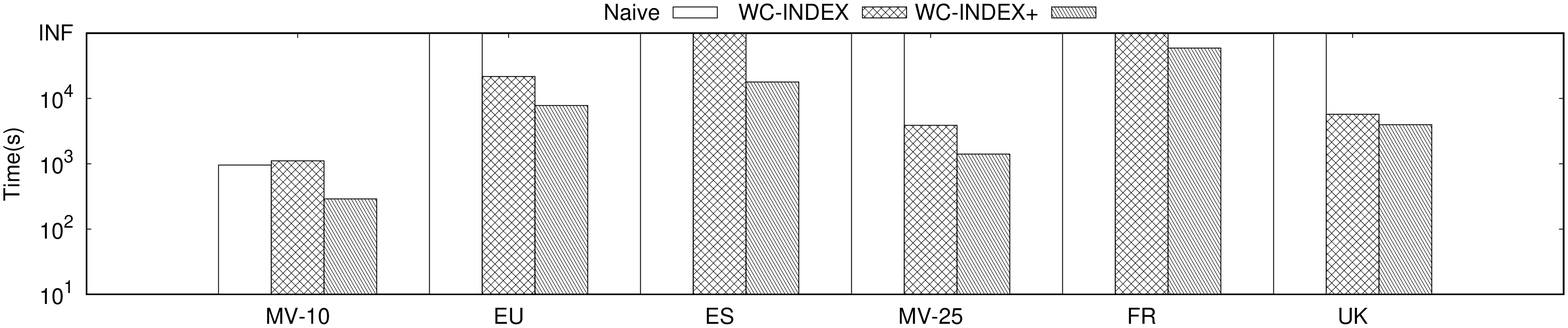}}
    \caption{Indexing Time (s) for baseline, \ours, and \ourp.}
    \label{fig:index_time_social}
        \vspace{-3mm}

\end{figure*}

\begin{figure*}[htbp]
    \centering
    {\includegraphics[width=0.8\linewidth]{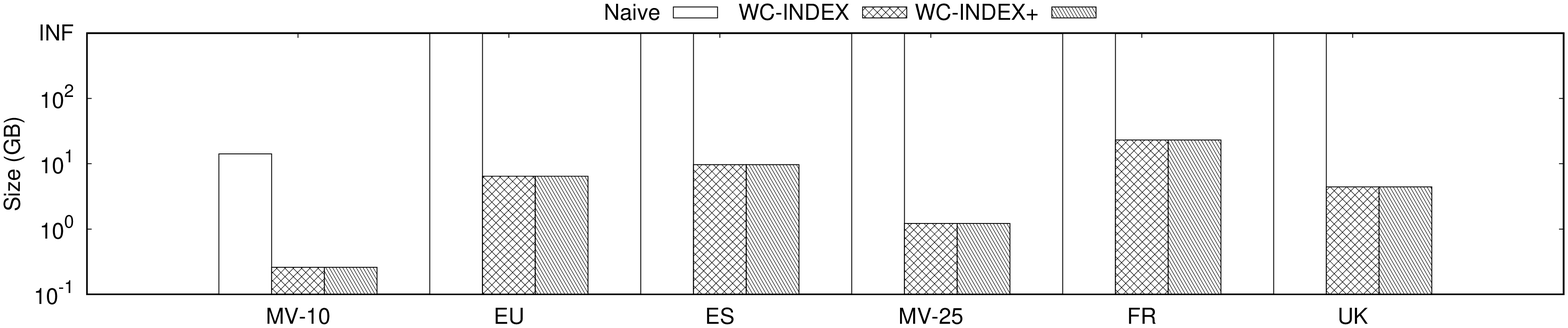}}
    \caption{Indexing Size (GB) for baseline, \ours, and \ourp.}
    \label{fig:index_size_social}
        \vspace{-3mm}

\end{figure*}

\begin{figure*}[htbp]
    \centering
    {\includegraphics[width=0.8\linewidth]{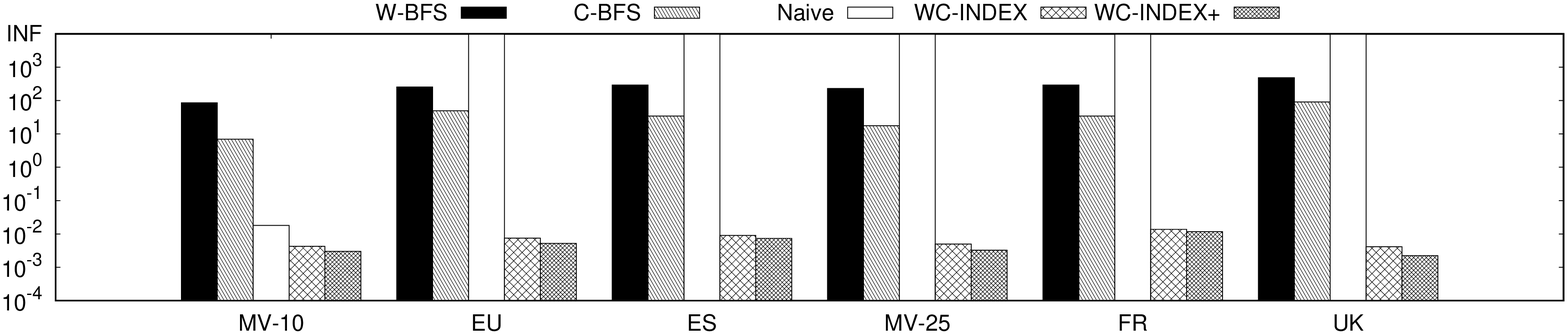}}
    \caption{Querying time (ms) for baselines, \ours, and \ourp.}
    \label{fig:query_time_social}
        \vspace{-3mm}

\end{figure*}

\noindent \textbf{Exp 5: Indexing Time, Size, and Query Time for Social Networks}.~Exp 5 evaluates the indexing time, size and query time for social networks. As shown in Figures~\ref{fig:index_time_social},~\ref{fig:index_size_social}, and~\ref{fig:query_time_social}, the patterns resemble those of road networks. It is interesting to notice that the indexing time and size over social networks are larger than that of road networks since social networks have a higher average degree. For the query time, this experiment does not consider the Dijkstra since the edge is unweighted and thus it is the same as W-BFS in the social networks. The query times of \ours, and \ourp \ are much faster than that of Na\"ive method.
\rev{Table \ref{tab:size_of_social_networks} summarizes the memory usage of storing the social networks.}

\begin{table}[thb]
    \centering
    \small
    \topcaption{Size of Social Networks} \label{tab:size_of_social_networks}
    \scalebox{0.96}{
    \begin{tabular}{|c||c|c|} \hline
        \textbf{Name} & \textbf{Dataset} & $|E(G)|$ \\ \hline
        MV-10 & Movielens-10m & 0.093\\ \hline
        EU & eu-2005 & 0.15\\ \hline
        ES & eswiki-2013 & 0.21\\ \hline
        MV-25 & Movielens-25m & 0.23 \\ \hline
        FR & frwiki & 0.29  \\ \hline
        UK & uk-2007 & 0.34 \\ \hline
        SO-Y& Stackoverflow (year) & 0.26\\ \hline
    \end{tabular}}
\end{table}





\section{Related Work}
\label{sect:related}

\stitle{Weight Constrained Shortest Path.} Given a directed graph $G$, and two vertices $s,t\in V(G)$, the WCSP aims to find a path $p$ between $s$ and $t$ such that the cost of $p$ (i.e., $c(p) = \sum_{e\in p} c(e)$) is minimized and the quality of $p$ (i.e., $w(p) = \sum_{e\in p} w(e)$) is less than a given threshold $W$. This problem is proved to be NP-hard \cite{smith2012solving,dumitrescu2003improved,qin2020software}. Our quality constraint shortest distance problem is inherently different as the constraint is imposed over each individual edge. 

\stitle{Label Constraint Shortest Path.} Label-constraint shortest path \cite{bonchi2014distance} returns the length of a shortest path over all paths that satisfy the predefined label sets (e.g., all labels of the edges in the path belong to the predefined label set). Likewise, there are some variants of this problem, e.g., Language constrained shortest path \cite{rice2010graph}, regular language constrained shortest path \cite{rice2010graph}, and some others in \cite{zhang2019correlation,shi2022indexing}. Nevertheless, they are different from our quality constraints. Thus, both weight and label constrained shortest path algorithms are not considered in this paper.

\noindent \textbf{Graph Search for Distance Queries.}~Both breadth-first search and Dijkstra's algorithm are classic algorithms for shortest path problems. Instead of Dijkstra's algorithm, the ALT algorithm~\cite{goldberg2005computing} employs A* search with a \textit{landmark}-based heuristic to speed up query processing. The notion of \textit{vertex reach} is proposed to reduce the search space for Dijkstra's algorithm in~\cite{gutman2004reach}. In the approaches that are based on \textit{arc}-flag~\cite{hilger2009fast}, a graph is partitioned into $k$ regions and each arc $(u,v)$ is associated with a $k$-bit flag of which the i-th bit indicates if there is a shortest path from $u$ to the $i$-th region via $(u,v)$. Based on the arc-flags, the search space of Dijkstra's algorithm can also be greatly reduced. The notation of \textit{highway hierarchy} (HH)~\cite{sanders2005highway} is designed to capture the natural hierarchy of road networks so that queries can be answered by searching the sparse high levels of HH, reducing the search space. 
\cite{geisberger2008contraction} introduced \textit{contraction hierarchy} (CH), in which, different from HH, each level consists of only one vertex. Its efficiency relies heavily on the noiton of \textit{shortcut}, which is to presere the distance between vertices after less important vertices are removed.

\noindent \textbf{Hub Labeling for Distance Queries.}~Another important class of algorithms for distance evaluation is hub labeling~\cite{cohen2003reachability}. In this class, a label $L(v)$ is computed for each vertex $v$ such that the distance between two vertices $s$ and $t$ can be obtained by inspecting $L(s)$ and $L(t)$ only, without searching the graph. In general, it is NP-hard to construct a labeling with the minimum size~\cite{cohen2003reachability}. In~\cite{abraham2011hub,abraham2012hierarchical}, efficient hub labelings for road networks are discussed. A labeling scheme that instead uses paths as hubs is presented in~\cite{akiba2014fast}. 

In~\cite{ouyang2018hierarchy}, under the assumption of small treewidth and bounded tree height, a scheme combining both hub labeling and hierarchy is proposed for road networks. For real graphs that are scale-free, pruned landmark labeling (PLL)~\cite{akiba2013fast} is the state-of-the-art and many extensions have been devised. For example, an external algorithm generating the same set of labels is proposed in~\cite{jiang2014hop}; a parallel algorithm is devised in~\cite{li2019scaling}; and~\cite{akiba2014dynamic} shows an algorithm to update the labels when new edges are inserted into the graph. 



\section{Conclusion and Future Works}
\label{sect:conclusion}
The shortest path is a fundamental concept in graph analytics. Existing works mainly focus on the distance computer of shortest paths. Nevertheless, finding a shortest path between $s$ and $t$ with a quality constraint along each edge is an important problem in many applications. To bridge this research gap, this paper presents a 2-hop labeling based solution to answer quality constrained shortest distance queries. Our techniques support query processing over large-scale graphs in real-time. 

\noindent \textbf{Future Works}.~
\rev{\textbf{re-indexing/dynamic index.} Effectively managing the index on dynamic graphs can be a future research direction. Here, we point out a plausible direction for extending our index to a dynamic scene. To handle edge insertion and deletion, a set of affected vertices can be computed and updates in the
index can be performed only on affected entries cause by the edge insertion/deletion. How to effectively compute affected vertices will be the focus of future research. Potential solutions are to utilize existing index entries instead of conducting constrained BFS for the edge inserted/deleted.}

\balance

\bibliographystyle{ieeetr}
\bibliography{cit}

\end{document}